%% file: arxiv2.tex
\documentclass[amsmath,onecolumn, floatfix,nofootinbib, prx, aps, usenames, dvipsnames,11pt]{quantumarticle}
\pdfoutput=1

\input{macros}

\begin{document}

\title{Composable security in relativistic\\ quantum cryptography}

\author{V. Vilasini}
\affiliation{Department of Mathematics, University of York, Heslington, York, YO10 5DD, UK}
\affiliation{Institute for Theoretical Physics, ETH Z\"{u}rich, 8093 Z\"{u}rich, Switzerland}
\email{vilasini.phys@gmail.com}

\author{Christopher Portmann}
\affiliation{Department of Computer Science, ETH Z\"{u}rich, 8092 Z\"{u}rich, Switzerland}
\email{chportma@ethz.ch}

\author{Lídia del Rio}
\affiliation{Institute for Theoretical Physics, ETH Z\"{u}rich, 8093 Z\"{u}rich, Switzerland}
\email{lidia@phys.ethz.ch}

\date{}

\begin{abstract}
Relativistic protocols have been proposed to overcome some impossibility results in classical and quantum cryptography. In such a setting, one takes the location of honest players into account, and uses the fact that information cannot travel faster than the speed of light to limit the abilities of dishonest agents. For example, various relativistic bit commitment protocols have been proposed  \cite{Kent2012,Lunghi2015}. Although it has been shown that bit commitment is sufficient to construct oblivious transfer \cite{Unruh2010}, composing specific relativistic protocols in this way is known to be insecure \cite[Appendix~A]{Kaniewski2015}. A composable framework is required to perform such a modular security analysis, but no known frameworks can handle models of computation in Minkowski space.

By instantiating the systems model from the Abstract Cryptography framework \cite{Maurer2011} with Causal Boxes \cite{Portmann2017}, we obtain such a composable framework, in which messages are assigned a location in Minkowski space (or superpositions thereof). This allows us to analyse relativistic protocols and to derive novel possibility and impossibility results. We show that (1) coin flipping can be constructed from the primitive channel with delay, (2) biased coin flipping,  bit commitment and channel with delay are all impossible without further assumptions, and (3) it is impossible to improve a channel with delay. Note that the impossibility results also hold in the computational and bounded storage settings. This implies in particular non-composability of all proposed relativistic bit commitment protocols, of bit commitment in the bounded storage model \cite{Damgard}, and of biased coin flipping \cite{Chailloux2013}.
\end{abstract}

\maketitle

\input{toc}

\newpage 

\section{Introduction}
\label{sec:introduction}
\input{introduction}

\section{Framework}
\label{sec:framework}
\input{framework}

\section{Results}
\label{sec:results}
\input{results}

\section{Discussion}
\label{sec:discussion}
\input{discussion}

\begin{acknowledgements}
We thank Renato Renner for  discussions on security definitions. 
VV acknowledges support from  FQXi for the funding to present this work at Oxford Quantum Networks 2017, the ETH Masters Scholarship from ETH Zurich, Switzerland and the Inlaks Scholarship from Inlaks Shivdasani Foundation, Mumbai, India for funding tuition and living expenses during her Masters.
CP acknowledges support from the Zurich Information Security
and Privacy Center.
LdR acknowledges support from the Swiss
National Science Foundation through 
SNSF project No.\ $200020\_165843$ and through the 
the National Centre of
Competence in Research \emph{Quantum Science and Technology}
(QSIT), and from  the FQXi grant \emph{Physics of the observer}.
\end{acknowledgements}



\newpage 
\appendix

\section*{\textsc{Appendix}}

\section{The causal box framework}
\label{appendix:causalboxes}

\input{appendix_causal_boxes}

\section{Proofs of all results}
\label{appendix:proofs}
\input{appendix_proofs}

\section{Unfair resources}
\label{appendix:extraresults}
\input{appendix_extra_results}

\vspace{2cm}


\input{bibcrypto.bbl} 

\end{document}

%% file: macros.tex
\usepackage[utf8]{inputenc}
\usepackage[english]{babel}
\usepackage[T1]{fontenc}

\usepackage[usenames,dvipsnames]{xcolor}
\usepackage{graphicx}
\usepackage[ colorlinks = true, 
             linkcolor = MidnightBlue,
             urlcolor  = MidnightBlue,
             citecolor = orange,
             anchorcolor = ForestGreen,
]{hyperref} 
\usepackage{float}

\usepackage[sans]{dsfont} 
\usepackage{bbm}
\usepackage[mathscr]{euscript}
\usepackage{cancel} 

\usepackage{tcolorbox}
\tcbuselibrary{theorems} 

\usepackage{amsmath}
\usepackage{amssymb}
\usepackage{amsthm}
\usepackage{units}
\usepackage{framed}
\usepackage{mdframed}
\usepackage{thmtools, thm-restate} 
\usepackage{ucs}
\usepackage{subcaption}
\usepackage{fullpage}
\usepackage{xspace}

\usepackage{multicol}
\usepackage{enumerate}
\usepackage[inline]{enumitem}

\usepackage{tikz}
\usetikzlibrary{arrows}
\usepackage{rotating, xcolor}
\usetikzlibrary{shapes}
\usetikzlibrary{hobby}
\usetikzlibrary{decorations.markings}

\makeatletter
\def\l@subsubsection#1#2{}
\def\l@subsection#1#2{}
\makeatother

\tcbset{colback=orange!5,colframe=orange!75!yellow, 
fonttitle= \bfseries, float=htb}

\newtcbtheorem[]{bigexample}{Example}{float*=htb, width=\textwidth}{ex}

\newtcolorbox[blend into=figures]{boxfigure}[2][]
{ float*=htb,width=\textwidth,lower separated=false, center upper, 
center title,title={#2},every float=\centering,#1}

\declaretheoremstyle[shaded={rulecolor=gray,rulewidth=1pt, bgcolor={rgb}{1,1,1}}]{boxed}

\declaretheoremstyle[spacebelow=\parsep,
    spaceabove=\parsep]{unboxed}

\declaretheoremstyle[shaded={rulecolor=MidnightBlue,rulewidth=1pt, bgcolor={rgb}{1,1,1}}]{boxedColour}

\declaretheoremstyle[shaded={rulecolor=Thistle,rulewidth=1pt, bgcolor={rgb}{1,1,1}}]{secboxed}

\declaretheoremstyle[shaded={rulecolor=YellowOrange,rulewidth=1pt, bgcolor={rgb}{1,1,1}}]{terboxed}

\declaretheoremstyle[shaded={rulecolor=Green,rulewidth=1pt, bgcolor={rgb}{1,1,1}}]{tetraboxed}

\declaretheorem[style=boxed]{lemma}
\declaretheorem[style=unboxed,sibling=lemma ]{remark}
\declaretheorem[sibling=lemma, style=boxed]{corollary}

\declaretheorem[style=boxed]{definition}

	\newcommand{\red}[1]{\textcolor{Red}{#1}}

	\newcommand{\fref}[1]{Fig.~\ref{#1}}

    \newcommand{\ket}[1]{\vert  #1 \rangle}
    \newcommand{\bra}[1]{\langle #1 |}

\newcommand{\cT}{\mathcal{T}}

\newcommand{\BC}{\ensuremath{\mathcal{BC}}}

\newcommand{\CF}{\ensuremath{\mathcal{CF}}}

\newcommand{\CD}{\ensuremath{\mathcal{CD}}}

	\renewcommand{\vec}[1]{\mathbf{#1}}

	\newcommand*{\eps}{\varepsilon}

\newcommand{\safezone}{trusted region\xspace}

\newcommand{\footnoteremember}[2]{\footnote{\label{#1}#2}\newcounter{#1}\setcounter{#1}{\value{footnote}}}
\newcommand{\footnoterecall}[1]{\hyperref[#1]{\footnotemark[\value{#1}]}}

%% file: toc.tex
\section*{{Contents}}
\begin{center}

\begin{tabular}{p{0.4\textwidth}  p{0.4\textwidth} }

    \textbf{Section \ref{sec:introduction}.} Introduction
    & \textbf{Appendix \ref{appendix:causalboxes}.} The causal box framework \\
    \textbf{Section \ref{sec:framework}.} Framework
    & \textbf{Appendix \ref{appendix:proofs}.} Proofs\\
    \textbf{Section \ref{sec:results}.} Results
    & \textbf{Appendix \ref{appendix:extraresults}.} Unfair resources \\
     \textbf{Section \ref{sec:discussion}.}  Discussion

\end{tabular}

\end{center}




%% file: introduction.tex
\paragraph{What this paper is about.}
We address construction of resources\footnote{The Abstract Cryptography framework \cite{Maurer2011} views cryptography as a resource theory: a protocol constructs a resource (e.g.\ a system that produces a random coin flip) from some other resource (e.g.\
a system that allows bit commitment). In the Universal Composability framework  \cite{Canetti}, resources correspond to ideal functionalities.} (e.g., an ideal coin flip or bit commitment) in relativistic quantum cryptography, and security definitions that are robust under composition of constructions. We prove new constructibility and impossibility results. By ``relativistic'' we mean basic special relativity: Minkowski space-time with limited signalling speed.

\paragraph{A cryptographic resource: bit commitment.} To illustrate the need for a composable analysis of relativistic quantum cryptography, we focus on bit commitment protocols, which have attracted interest in recent years \cite{Kent1999,Kent2012,Kaniewski2013,Lunghi2015}. Bit commitment is a crucial cryptographic primitive, from which we can construct oblivious transfer\footnoteremember{fn:bc}{Constructing oblivious transfer (and thus multi-party computation) from bit commitment holds only in the quantum case \cite{Unruh2010}.} \cite{Unruh2010}, multi-party computation\footnoterecall{fn:bc} \cite{Unruh2010}, coin flipping \cite{Blum1983}, and zero-knowledge proofs \cite{Kilian1988}.

A bit commitment protocol (\BC) between two players (say Alice and Bob) typically involves two phases. In the \emph{commit phase}, Alice commits to a bit $a \in \{0,1\}$ with Bob by exchanging information with him. 
In the \emph{open phase}, Alice chooses to open her commitment to Bob and reveals her bit to him through an exchange of information. Intuitively speaking, security of bit commitment has two requirements: 
 \begin{description}
     \item[Hiding:] when Alice is honest, Bob has no information about $a$ before the open phase.
    \item[Binding:] when Bob is honest,  Alice must not be able to change the value of $a$ between the commit and open phases without him detecting her malicious behavior.
 \end{description} 

These requirements can be formalized under different security definitions. Not all models of security of \BC\ are  composable: for example the $\epsilon$-weakly binding definition of \cite{Kaniewski2013} is not. There, Alice is allowed to commit to a bit without knowing its value, which if used as a subroutine in a coin flipping protocol, would allow dishonest players to perfectly correlate the coin flips from different coins. Similar weaknesses in current definitions of relativistic bit commitment have been exploited to show that using these protocols as subroutine in a larger cryptosystem is insecure~\cite[Appendix A]{Kaniewski2015}. In this work, we model security such that the constructed \BC\ resource can be securely used in arbitrary context. Let us first review some known results.

\paragraph{Impossiblity of classical bit commitment.}
In 2001, Canetti and Fischlin showed that constructing a \BC\ resource without any setup assumptions is impossible \cite{Canetti2001}. They proved this for a classical non-relativistic setting through a classical man-in-the-middle attack (MITM).
Consider a cheating Alice simultaneously running two \BC\ protocols: one with Bob, in which she is the committer, and one with Charlie, in which she is the receiver. She can commit to Charlie's bit with Bob by simply forwarding their messages to each other during the commit phase. 
Note that the proof from \cite{Canetti2001} is restricted to the classical setting, and does not imply the impossibility of constructing a \BC\ resource in either quantum or relativistic settings.

\paragraph{Impossibility of quantum bit commitment.}
Using a stand-alone definition with information-theoretic security, Mayers, and Lo and Chau \cite{Mayers1997,Lo1997,Lo1998} independently showed between 1996 and 1997 that no secure quantum bit commitment protocol can be constructed without further assumptions (for example regarding the operations that (dishonest) parties can perform on their systems), because due to Uhlmann's theorem, if Bob cannot distinguish between the commitment to a $0$ or a $1$, then there exists a unitary on Alice's system allowing her to change the commitment from $0$ to $1$. Possibility results are obtained by restricting the adversary's capabilities. For example, Unruh showed in \cite{Unr11} that if the adversary has bounded quantum memory, bit commitment that is composable in certain restricted settings is possible.\footnote{The model used in \cite{Unr11} does not guarantee security when a protocol is composed with itself. There is thus no contradiction with the impossibility proof for bit commitment in the bounded storage model in this work, which shows that any bit commitment protocol run in parallel with another instance if itself is insecure.} In \cite{Unr13} Unruh also shows that everlasting quantum bit commitment is achievable, if we assume signature cards as trusted setup.

\paragraph{Relativistic protocols.}
In the hope of avoiding such attacks without making unproven assumptions on the adversary's capabilities, one turns to relativistic protocols and imposes relativistic causal constraints on agents located in Minkowski space---no-signalling between space-like separated agents and a maximum propagation speed for signals.
An example is Kent’s 2012  relativistic \BC\ protocol \cite{Kent2012}, which is immune to the Mayers-Lo-Chau attack, since the sender splits into two space-like separated agents who can no longer perform suitable unitaries on their joint systems. Like other relativistic \BC\ protocols, this protocol implements a timed commitment which is secure only within a time window given by  the time taken by light to travel between remote agents. However, it only satisfies a non-composable, weakly-binding security definition \cite{Kaniewski2013}. As we will see, this protocol is susceptible to a man-in-the-middle attack and therefore cannot be securely run as a subroutine in arbitrary protocols. 

\paragraph{Composability of relativistic protocols.}
In relativistic settings, the existing negative results  are obtained by analyzing specific examples of protocols and attacks where composition fails \cite{Kaniewski2013,Kaniewski2015}. However, without an overall coherent framework for modelling composability in relativistic cryptography, it is impossible to obtain general positive and negative results.

\subsection{Overview and scope of our results}

In this work we introduce a framework for modelling composable cryptographic security in the presence of classical, quantum and no-signalling adversaries, and apply it to prove new positive and negative results in relativistic quantum cryptography.
We do this by modelling the abstract information-processing systems of the \emph{Abstract Cryptography} framework \cite{Maurer2011} as \emph{Causal Boxes} \cite{Portmann2017}, which we instantiate with Minkowski space-time.
Our framework can also be applied to situations where agents exchange a superposition of different numbers of messages in a superposition of orders in time, and provides an operational formalism for studying indefinite causal structures. 

We analyse three cryptographic resources, defined in Section~\ref{sec:framework}. Coin flipping (\CF, including biased variations) and bit commitment (\BC) are standard in the composable security literature, though in this work our formalization involves space-time---inputs and outputs are produced at certain locations in Minkowski space. We also introduce a channel with delay (\CD), which is motivated by the fact that in relativistic bit commitment protocols, the commitment is automatically opened after some (predefined) time, thus resembling a \CD\ more than a \BC.\footnote{There may be different ways of modeling a relativistic bit commitment resource, e.g., the committer may have the option of aborting before the commitment is opened, see the discussion in Sec.~\ref{sssec: CD}.} The following results are summarized \fref{fig:results}. 

\begin{figure}[tb]
\centering
	\begin{tikzpicture}[line width=0.20mm, scale=0.8, transform shape]
		\node[draw,text width=1.5cm,minimum height=2cm,minimum width=2cm,text centered] (BC) at (1,1) {Bit commitment};
		\draw [arrows={-latex'}] (2,1)--node[midway,above]{\cite{Unruh2010}}(3,1); 
		\draw (3,0) rectangle node[rectangle split,rectangle split parts=2]{Oblivious \nodepart{second} transfer} (5,2);
		\draw [arrows={-latex'}] (5,1)--node[midway,above]{\cite{Unruh2010}}(6,1); \draw (6,0) rectangle node[rectangle split,rectangle split parts=3]{Multi \nodepart{second} party \nodepart{third} computations} (8,2);
		\draw [arrows={-latex'}] (1,0)--node[midway,swap,auto]{\cite{Demay2013}}(1,-1);
		\node[draw,text width=1.5cm,minimum height=2cm,minimum width=2cm,text centered] (CF) at (1,-2) {Coin flipping};
        \draw [arrows={-latex'}] (0,1)--node[midway,above]{\cite{Canetti2001}}(-1,1); \draw (-3,0) rectangle node[rectangle split,rectangle split parts=3]{Zero \nodepart{second} knowledge \nodepart{third} proofs} (-1,2);
		\draw [arrows={-latex'}, blue] (1,-4)--node[auto] {\textbf{Theorem~\ref{claim: constructibility}}} (1,-3);
		\node[draw,blue,text width=1.5cm,minimum height=2cm,minimum width=2cm,text centered] (CD) at (1,-5) {Channel with delay};
	    \node[draw,text width=1.8cm,minimum height=2cm,minimum width=2cm,text centered] (channel) at (5.5,-2) {Direct communication};

		\draw[thick, red] (2.75,-1)--(3.75,0); \draw[thick, red] (3.75,-1)--(2.75,0);
		\draw[thick, red] (2.75,-2.5)--(3.75,-1.5); \draw[thick, red] (3.75,-2.5)--(2.75,-1.5); 
		\draw[thick, red] (2.75,-4)--(3.75,-3); \draw[thick, red] (3.75,-4)--(2.75,-3);
		\draw[thick, red] (-2.5,-5.5)--(-1.5,-4.5); \draw[thick, red] (-1.5,-5.5)--(-2.5,-4.5);
		\draw [arrows={-latex'}, blue] (-1.2,-5)--(-3,-5); 
		\draw[blue] (-5,-6) rectangle node[blue, rectangle split,rectangle split parts=4]{Channel \nodepart{second} with \nodepart{third} 
	    larger	\nodepart{fourth} delay} (-3,-4);
		\node[align=center, blue] at (-0.7,-5) {\huge{$\mathbf{\times n}$}}; 
		
		\node[align=center, red] at (-2,-4.2) {\bf Theorem~\ref{claim: impos2}};
        \draw[arrows={-latex'}, blue] (channel) to (BC); \node[align=center, red] at (4.75,-0.5) {\textbf{Corollary~\ref{claim: bcimpos}}};
        \draw[arrows={-latex'}, blue] (channel) to (CF); \node[align=center, red] at (3.25,-1.2) {\textbf{Theorem~\ref{claim: impossibility}}};
        \draw[arrows={-latex'}, blue] (channel) to (CD); \node[align=center, red] at (4.75,-3.5) {\textbf{Corollary~\ref{claim: cdimpos}}};
		
	\end{tikzpicture}
	\caption{{\bf Summary of our results.} We assume Minkowski space-time with limited speed of signalling (upper bounded by the speed of light in vacuum, $c$). Existing results are represented in black and new results (obtained in this paper), in blue and red.  
	An arrow $\mathcal R \to \mathcal S$ means that it is possible to construct resource $\mathcal S$ from resource $\mathcal R$. When the arrow is crossed, that means that no such construction exists.}
	\label{fig:results}
\end{figure}
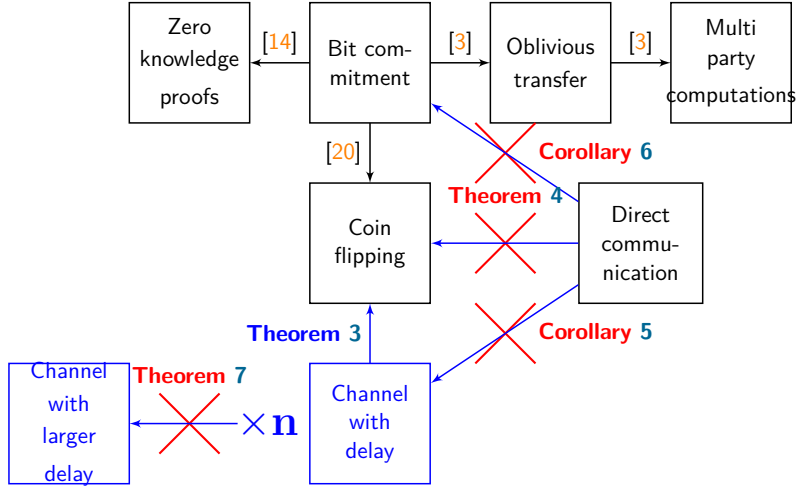

\paragraph{Constructibility results.}
We show that an  unbiased coin flipping resource \CF\ can be constructed from a channel with delay resource, \CD\  (Theorem~\ref{claim: constructibility}). For comparison, Blum's protocol \cite{Blum1983},  constructs  a weaker, biased\footnote{Originally, Blum's protocol constructs an \emph{unfair} coin flip, in which one party can abort after seeing the flip \cite{Blum1983}. This may be transformed into a \emph{biased} coin flip if the honest party flips a coin locally when the dishonest party aborts \cite{Demay2013}.} coin flipping resource from a bit commitment resource \cite{Demay2013}.  We provide an explicit protocol to construct \CF\ from \CD\ and prove its security. The proof holds even in the presence of adversaries that are not bounded by quantum physics, but only non-signalling constraints.

\paragraph{Impossibility results.}
In Theorem~\ref{claim: impossibility} we show that constructing a (biased) coin flipping resource is impossible in the relativistic setting without additional setup assumptions (e.g., the presence of a shared resource such as \CD). This result holds even if the players are only bounded by non-signalling constraints,\footnoteremember{fn:NS}{A non-signalling player can generate non-signalling correlations between their own trusted agents at different locations. Note however that that if we were to allow two distrusting players (Alice and Bob) to generate non-signalling correlations between them, this would have to be modeled as an extra setup assumption, namely a shared resource.} or if we restrict the adversary to being computationally bounded or having bounded storage. Impossibility of bit commitment follows from Blum's construction~\cite{Blum1983,Demay2013} of $\CF$ from $\BC$ (Corollary~\ref{claim: bcimpos}), and impossibility of constructing a channel with delay \CD\ follows from Theorem~\ref{claim: constructibility} (Corollary~\ref{claim: cdimpos}).

Since the literature on relativistic bit commitment also studies the case of extending the time during which such a commitment holds, we also look at the task of constructing a channel with a long delay $\CD_{\text{long}}$ from multiple channels (labelled by $i$) with shorter delays $\{\CD_{\text{short}}^i\}_i$. We show that this again is impossible without other setup assumptions than the assumed channels with delays $\{\CD_{\text{short}}^i\}_i$ (Theorem~\ref{claim: impos2}). This impossibility result holds irrespective of whether the protocol is classical, quantum or non-signalling,\footnoterecall{fn:NS} and also holds if the adversary is computationally bounded.

\paragraph{Consequences of these results.}
Many quantum protocols have been proposed in the relativistic setting to circumvent classical impossibility results for \BC. To the best of our knowledge, none of these protocols have been successfully used as subroutines in larger cryptosystems (which is the main motivation for constructing such primitives), and attempts to do so are insecure~\cite[Appendix A]{Kaniewski2015}. But due to the lack of composable framework that can model Minkowski space, it has been impossible to prove whether composable constructions of these resources do exist. Our results show that allowing quantum (and even non-signalling\footnoterecall{fn:NS}) protocols that respect relativistic constraints is not sufficient to construct \BC, \CF, or \CD\ without additional assumptions. This implies that none of the proposed relativistic bit commitment schemes are composable (e.g., \cite{Kent1999,Kent2012,Kaniewski2013,Lunghi2015}). This extends to the non-relativistic setting (e.g., \cite{Chailloux2013}), since a non-realtivistic protocol corresponds to the special case where all players are in the same position in space (and thus do not have any constraints on the speed of communication). Our proof also holds against computationally bounded adversaries, and adversaries with bounded storage, which implies that results in the bounded storage model are not composable either (e.g., \cite{Damgard}).

The other problem considered in the literature on relativisitic bit commitment is extending the time of a commitment. Our results show that this cannot be done with a composable definition of timed commitment (see the definition of \CD\ in Sec.~\ref{ssec: resources} and following discussion). Hence the techniques used in \cite{Lunghi2015,Chakraborty2015} to extend the time of a relativistic bit commitment cannot be used in a composable way. As for previous results, this holds as well if the adversary is computationally limited or has bounded quantum memory.

The framework naturally allows positive results to be proven as well---by making extra setup assumptions. This approach was used by Unruh~\cite{Unr13}, who showed (everlasting) quantum bit commitment is achievable if we assume signature cards as trusted setup. In this work we construct a \CF\ resource from a \CD, and leave open the problem of finding weaker assumptions that still allow $\CF$ or $\BC$ to be constructed.

\subsection{Structure of this paper.}

In Sec.~\ref{sec:framework} we introduce the model that we use to prove our results. We explain the Abstract Cryptography framework in Sec.~\ref{ssec: AC}. We give an overview of Causal Boxes instantiated with Minkowski space in Sec.~\ref{ssec: CB}---a formal presentation of Causal Boxes is given in Appendix~\ref{appendix:causalboxes}. And in Sec.~\ref{ssec: resources} we define the two party resources $\CF$, $\CD$, and $\BC$. Our results are then presented in Sec.~\ref{sec:results} and the proofs are given in Appendix~\ref{appendix:proofs}. Finally, we conclude in Sec.~\ref{sec:discussion} with a discussion of these results.

%% file: framework.tex
\subsection{Composable security: the abstract cryptography framework \texorpdfstring{\cite{Maurer2011}}{ACcite}}
\label{ssec: AC}

\subsubsection{Resources, converters and distinguishers.}

Let us review the basics of the abstract cryptography framework. \cite{Maurer2011} 
The following is adapted from  \cite{Portmann2014} for the case of protocols between two mutually distrusting parties (e.g., bit commitment, coin flipping) and has been simplified for our purposes. We refer the reader to \cite{Maurer2011} and \cite{Portmann2014} for more general definitions and further examples. 

\paragraph{Abstract systems.} Abstract cryptography views cryptography as a resource theory: a protocol constructs a resource from some other resource, e.g., Blum's protocol \cite{Blum1983} constructs a coin flipping resource from a bit commitment resource. In this section we introduce the building blocks of the framework---resources, converters (e.g., protocols) and a notion of distance (distinguishability) between resources---and in Section~\ref{ssec: CB} we explain how these objects are instantiated with Causal Boxes~\cite{Portmann2017}.

A \emph{resource} $\mathcal R$ in a two party setting is an (abstract) system with interfaces $i \in \{A,B\}$, each accessible to a user $i$ (and their trusted agents) providing them with certain controls. An operation that is performed by a party at their interface is modeled as a \emph{converter}: a system $\alpha$ with an outside and an inside interface, the inner interface connects to an interface $i$ of the resource, and the outer interface becomes the new interface of the resulting resource. We write $\alpha_i \mathcal R$ to denote the resource resulting from connecting $\alpha$ to the $i$ interface of $\mathcal R$. This is illustrated in \fref{fig: converter}.

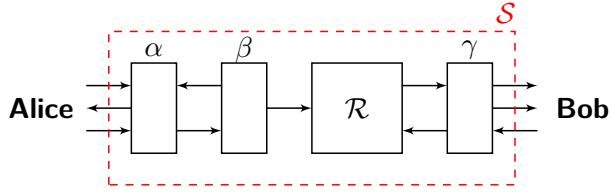
\begin{figure}[tb]
\centering
    \begin{tikzpicture}[line width=0.20mm, scale=0.6]
        \draw (0,0) rectangle (1,2); \draw (2,0) rectangle (3,2); \draw (4,0) rectangle node{$\mathcal{R}$}(6,2); \draw (7,0) rectangle (8,2); 
        \node[align=center] at (0.5,2.3) {$\alpha$}; \node[align=center] at (2.5,2.3) {$\beta$};
        \node[align=center] at (7.5,2.3) {$\gamma$};
        \draw[arrows={-latex'}] (-1,1.5)--(0,1.5); \draw[arrows={-latex'}] (0,1)--(-1,1); \draw[arrows={-latex'}] (-1,0.5)--(0,0.5);
        \draw[arrows={-latex'}] (2,1.5)--(1,1.5);
        \draw[arrows={-latex'}] (1,0.5)--(2,0.5);
        \draw[arrows={-latex'}] (3,1)--(4,1);
        \draw[arrows={-latex'}] (6,1.5)--(7,1.5);
        \draw[arrows={-latex'}] (7,0.5)--(6,0.5); \draw[arrows={-latex'}] (8,1.5)--(9,1.5);
        \draw[arrows={-latex'}] (8,1)--(9,1);
        \draw[arrows={-latex'}] (9,0.5)--(8,0.5); \node[align=center] at (-2,1) {\textbf{Alice}}; \node[align=center] at (10,1) {\textbf{Bob}}; \draw [dashed, red] (-0.5,-0.7) rectangle (8.5,2.7); \node[align=center, red] at (8.3,3.1) {\textbf{$\mathcal{S}$}};
    \end{tikzpicture}
    \caption{Starting from a resource $\mathcal R$, converters $\alpha, \beta$ and $\gamma$ construct a new resource $\mathcal S = \alpha_A\beta_A\gamma_B \mathcal R$. The sequences of arrows at the interfaces between objects represent (arbitrary) rounds of communication. For simplicity, we may omit the indices, $\mathcal S = \alpha \beta \mathcal R \gamma$, so that converters to the left of the resource ($\alpha, \beta$) are implicitly connected to Alice's interface, and converters on the right ($\gamma$) are connected to Bob's.}
    \label{fig: converter}
\end{figure}

 



\paragraph{Distinguishing resources.}
The security of a cryptographic system is quantified in terms of \emph{distinguishability} from a corresponding ideal system (Fig.~\ref{fig: dist}).  
For example, the ideal resource ``random bit generator'', $\mathcal{S}$, would be a black box that generates and outputs a uniformly random bit at a time $t$ which is independent of everything outside the box. A specific practical implementation $\mathcal R$ of this functionality could be a quantum protocol: prepare a qubit in a state $\frac1{\sqrt2}(\ket{0}+\ket1)$, measure it in the $Z$-basis and output the measurement result at time $t$. Treated as black boxes, both resources $\mathcal{R}$ and $\mathcal{S}$  output a uniformly random classical bit and cannot be distinguished by an outsider. 
For more complex resources, we may ask: distinguishability from whose perspective? 
Here, the traditional notion of an adversary is generalized to an arbitrary \emph{distinguisher} which models not only possible adversarial behaviour but also the whole environment of a cryptographic protocol. In other words, a distinguisher models information-processing steps that could take place before, after or during the protocol under consideration.

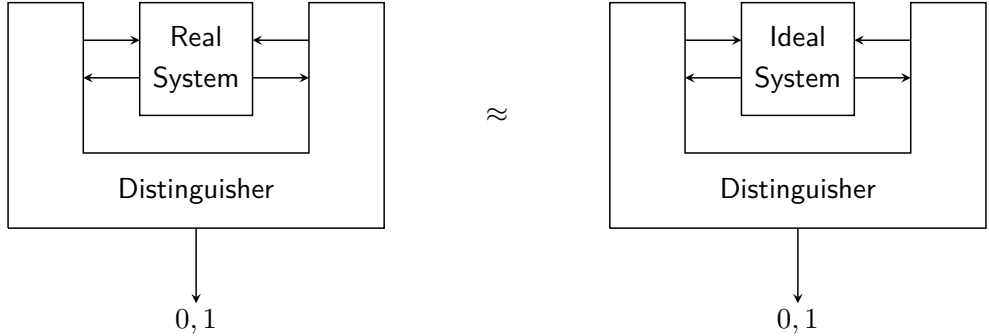
\begin{figure}[tb]
	\centering
	\begin{tikzpicture}[line width=0.20mm, scale=1.0, transform shape]
		\draw (0,0)--(0,3)--(1,3)--(1,1)--(4,1)--(4,3)--(5,3)--(5,0)--(0,0); \draw (8,0)--(8,3)--(9,3)--(9,1)--(12,1)--(12,3)--(13,3)--(13,0)--(8,0);
		\draw (1.75,1.5) rectangle node [rectangle split,rectangle split parts=2] {Real \nodepart{second} System} (3.25,3); 
		\draw (9.75,1.5) rectangle node [rectangle split,rectangle split parts=2] {Ideal \nodepart{second} System} (11.25,3);
		\node [align=center] at (2.5,0.5) {Distinguisher}; \node [align=center] at (10.5,0.5) {Distinguisher}; \node [align=center] at (6.5,1.5) {$\approx$};
		\draw [arrows={-stealth}] (1.75,2)--(1,2); \draw [arrows={-stealth}] (1,2.5)--(1.75,2.5); \draw [arrows={-stealth}] (3.25,2)--(4,2); \draw [arrows={-stealth}] (4,2.5)--(3.25,2.5);
		\draw [arrows={-stealth}] (9.75,2)--(9,2); \draw [arrows={-stealth}] (9,2.5)--(9.75,2.5); \draw [arrows={-stealth}] (11.25,2)--(12,2); \draw [arrows={-stealth}] (12,2.5)--(11.25,2.5);
		\draw [arrows={-stealth}] (2.5,0)--(2.5,-1); \draw [arrows={-stealth}] (10.5,0)--(10.5,-1);
		\node [align=center] at (2.5,-1.25) {$0,1$};
		\node [align=center] at (10.5,-1.25) {$0,1$};
	\end{tikzpicture}
	\caption{{\bf Security in terms of distinguishers.} Composable security of a real resource is defined in terms of the success probability of a class of distinguishers (for example computationally bounded or unbounded, classical, quantum or non-signalling) in distinguishing the real system from the ideal one.
	A distinguisher, modelling all the environment of a resource, is given black-box access to either the real or the ideal systems and a complete description of the input-output behaviour of both systems and must guess which one it was interacting with by outputting either a $0$ or a $1$. The distinguishing advantage is then given by the statistical distance between the two random bits output when interacting with the real and ideal systems, respectively.}
	\label{fig: dist}
\end{figure}

\begin{definition}[Distinguishing advantage \cite{Portmann2014}]
  \label{definition: dist adv} 
  A \emph{distinguisher} (Figure~\ref{fig: dist}) for two resources $\mathcal R, \mathcal S$ is a system $\mathcal{D}$ with two interfaces: an inside interface that connects to all the interfaces of a resource, $\mathcal R$ or $\mathcal S$, and an outside interface that outputs a single bit: a guess whether it is interacting with $\mathcal R$ or $\mathcal S$. The advantage of a specific distinguisher $\mathcal D$ is then given by \[d^{\mathcal{D}}(\mathcal{R}, \mathcal{S}) = \left| \Pr \left[ \mathcal D ( \mathcal R ) = 0 \right] - \Pr \left[ \mathcal D ( \mathcal S ) = 0 \right] \right|,\] where $\mathcal D(\mathcal R)$ is the output of $\mathcal D$ when interacting with $\mathcal R$.
  
  The distinguishing advantage for a class of distinguishers $\mathbb D$ is defined as
  \[ d^{\mathbbm D}(\mathcal R,\mathcal S) = \sup_{\mathcal D \in \mathbbm D} d^{\mathcal D}(\mathcal R,\mathcal S).\]
\end{definition}

The distinguishing advantage is a pseudo-metric on the space of resources satisfying the identity, symmetry and triangle inequality properties \cite{Portmann2014}. If a class of distinguishers $\mathbb D$ is such that for every $\mathcal D \in \mathbbm D$, $\mathcal D \alpha \in \mathbbm D$, then the pseudo-metric is non-increasing under application of the converter $\alpha$, i.e.\ $d^{\mathbbm{D}}(\alpha R, \alpha S) \leq d^{\mathbb{D}}(R, S)$.

\paragraph{Classes of distinguishers.}
Changing the power of the distinguisher (e.g., with some computational or memory bound, or performing only classical, quantum or non-signalling operations) results in different metrics and different levels of security. For example, if a protocol provides classical computational security, this means that the resource constructed may be perfectly indistinguishable from an ideal resource when considering only computationally bounded distinguishers, but they could be easily distinguished using computationally unbounded (or quantum) distinguishers. This is addressed in more detail in the following.

\subsubsection{Cryptographic security.}

We want to address questions such as ``does a protocol $\Pi$ construct the ideal resource $\mathcal S$ from an initial resource $\mathcal R$?'' The resource constructed will essentially depend on which players may be honest. For example, in the case of coin flipping, if both parties are honest we expect the protocol to construct a resource that provides each party with a copy of the same uniformly random bit. But if a party is dishonest, this might be a too strong requirement. Instead, we ``only'' construct a resource that allows the dishonest party to either abort if she does not like the value of the generated bit, or to bias the bit towards either $0$ or $1$. \cite{Demay2013}

In the case of two party protocols, we want to make a statement about three cases: where both parties are honest, Alice is dishonest, and Bob is dishonest. The resources available to the players are given by a tuple $(R, R_A, R_B)$, where $R$ denotes the shared resource when both a honest, $R_A$ is available to an honest Bob and dishonest Alice (presumably, providing more functionalities to Alice than $R$), and $R_B$ is shared between an honest Alice and dishonest Bob. Likewise, the constructed resources are also given by such a tuple $(S,S_A,S_B)$.

A two-player protocol $\Pi= (\Pi_A, \Pi_B)$ is essentially a pair of converters that can be connected to the interfaces of the shared resources $(R, R_A, R_B)$. When both are honest, the resulting system is given by $\Pi_A R\, \Pi_B$ (the ``real system''), which should be close to indistinguishable from the ideal resource $S$.

When Alice is dishonest, the protocol $\Pi_A$ is removed in the corresponding real system, because we do not know what protocol a dishonest player would follow. On the ``real'' side we now have $R_A \Pi_B$. On the ideal side, we have $S_A$, but in most cases $R_A \Pi_B$ and $S_A$ are trivially distinguishable since Alice's interface of $R_A \Pi_B$ is generally very different from her interface of $S_A$: $S_A$ provides an idealized interface, which, in the case of coin flipping, might allow Alice to abort. In the real system, $R_A \Pi_B$ Alice receives messages from Bob, and could provoke an abort by sending invalid messages or not responding.

To allow for the comparison and define security against dishonest Alice, we require the existence of a converter (or \emph{simulator}) $\sigma_A$ which when connected to Alice's interface of $S_A$ makes these two systems close to indistinguishable. Note that connecting this simulator $\sigma_A$  only makes Alice weaker, since any operation performed by the simulator could equivalently be performed by an adversary connected directly to the interface of the ideal resource. Further, the simulator's behaviour is independent of the internal workings of the ideal functionality $S_A$. Security in the case of a dishonest Bob is defined similarly.

\begin{definition}[Cryptographic security \cite{Portmann2014}] 
\label{definition: security}
A protocol $\Pi=(\Pi_A,\Pi_B)$ constructs $\mathcal{S} = (S, S_A, S_B)$ from $\mathcal{R} = (R, R_A, R_B)$ within a distance $\epsilon$, 
with respect to a set $\mathbbm D$ of distinguishers and a set $\mathbbm S \ni \Pi_A, \Pi_B$ of converters, 
if the following conditions hold:
\begin{align*}
    & d^{\mathbbm D}(\Pi_A R \Pi_B, S)\leq \eps, \\
    \exists\, \sigma_A \in \mathbbm S, \quad  & d^{\mathbbm D}( R_A \Pi_B, \sigma_A S_A)\leq \eps, \\
    \exists\, \sigma_B \in \mathbbm S,  \quad & d^{\mathbbm D}( \Pi_A R_A, S_B \sigma_B)\leq\eps.
\end{align*}
We sometimes write $\mathcal R \xrightarrow{\Pi} \mathcal S$ to denote such a constructions. These conditions are illustrated in Fig.~\ref{fig: security}.
\end{definition}

\begin{figure}[tbp]
	\begin{subfigure}{1.0\textwidth}
		\centering
		\begin{tikzpicture}[line width=0.20mm, scale=0.55]
			\draw [white] (-1,0) rectangle (18, 5); \draw (0,0) rectangle (1,4); \draw (4,0) rectangle (5,4); \draw (2,2.5) rectangle node{$R$} (3,3.5); \draw (11,1) rectangle node{$S$} (13, 3);
			\draw [arrows={-stealth}] (1,2.83)--(2,2.83); \draw [arrows={-stealth}] (2, 3.166)--(1,3.166); \draw [arrows={-stealth}] (4,2.83)--(3,2.83); \draw [arrows={-stealth}] (3, 3.166)--(4,3.166);
			\draw [arrows={-stealth}] (1,0.5)--(4,0.5); \draw [arrows={-stealth}] (4,0.83)--(1,0.83); \draw [arrows={-stealth}] (1,1.16)--(4,1.16);
			\draw [arrows={-stealth}] (-1,1.33)--(0,1.33); \draw [arrows={-stealth}] (0,2.66)--(-1,2.66); \draw [arrows={-stealth}] (6,1.33)--(5,1.33); \draw [arrows={-stealth}] (5,2.66)--(6,2.66);
			\draw [arrows={-stealth}] (10,1.66)--(11,1.66); \draw [arrows={-stealth}] (11,2.32)--(10,2.32); \draw [arrows={-stealth}] (14,1.66)--(13,1.66); \draw [arrows={-stealth}] (13,2.32)--(14,2.32);
			\node [align=center] at (0.5,4.5) {$\Pi_A$};
			\node [align=center] at (4.5,4.5) {$\Pi_B$};
			\node [align=center] at (8,2) {\large{$\approx_{\epsilon}$}};
		\end{tikzpicture}
		
        $$d^{\mathbbm D}(\Pi_A \, R\, \Pi_B, S)\leq \eps$$		

		\caption{When both parties are honest, the composition of Alice's and Bob's protocol with their shared resource must be $\eps$-indistinguishable from the constructed resource $S$.}
	\end{subfigure}
	
	\begin{subfigure}{1.0\textwidth}
		\centering
		\begin{tikzpicture}[line width=0.20mm, scale=0.55]
			\draw [white] (-1,0) rectangle (18, 5); \draw (4,0) rectangle (5,4); \draw (2,2.5) rectangle node{$R_A$} (3,3.5); \draw (13,1) rectangle node{$S_A$} (15, 3); \draw [red] (11,1) rectangle (12,3);
			\draw [arrows={-stealth}] (1,2.83)--(2,2.83); \draw [arrows={-stealth}] (2, 3.166)--(1,3.166); \draw [arrows={-stealth}] (4,2.83)--(3,2.83); \draw [arrows={-stealth}] (3, 3.166)--(4,3.166);
			\draw [arrows={-stealth}] (1,0.5)--(4,0.5); \draw [arrows={-stealth}] (4,0.83)--(1,0.83); \draw [arrows={-stealth}] (1,1.16)--(4,1.16);
			\draw [arrows={-stealth}] (6,1.33)--(5,1.33); \draw [arrows={-stealth}] (5,2.66)--(6,2.66);
			\draw [arrows={-stealth}] (12,1.66)--(13,1.66); \draw [arrows={-stealth}] (13,2.32)--(12,2.32); \draw [arrows={-stealth}] (16,1.66)--(15,1.66); \draw [arrows={-stealth}] (15,2.32)--(16,2.32);
			\draw [arrows={-stealth}, red] (11,2.665)--(10,2.665); \draw [arrows={-stealth}, red] (10,2.332)--(11,2.332); \draw [arrows={-stealth}, red] (10,1.999)--(11,1.999);
			\draw [arrows={-stealth}, red] (11,1.666)--(10,1.666); \draw [arrows={-stealth}, red] (10,1.333)--(11,1.333);
			\node [align=center] at (4.5,4.5) {$\Pi_B$};
			\node [align=center, red] at (11.5,3.5) {$\sigma_A$};
			\node [align=center] at (8,2) {\large{$\approx_{\epsilon}$}};
		\end{tikzpicture}
		
		$$ \red{\exists\, \sigma_A \in \mathbbm S}, \quad  d^{\mathbbm D}( R_A\, \Pi_B, \red{\sigma_A} \, S_A)\leq \eps$$
		
		\caption{When Alice is dishonest and Bob is honest, the resulting real system obtained by removing Alice's honest protocol must be $\eps$-simulatable  by connecting a converter   $\sigma_A$ (called a  \emph{simulator}) to Alice's interface of  corresponding ideal system, $S_A$.}
	\end{subfigure}
	
	\begin{subfigure}{1.0\textwidth}
		\centering
		\begin{tikzpicture}[line width=0.20mm, scale=0.55]
			\draw [white] (-1,0) rectangle (18, 5); \draw (0,0) rectangle (1,4); \draw (2,2.5) rectangle node{$R_B$} (3,3.5); \draw (11,1) rectangle node{$S_B$} (13, 3); \draw [red] (14,1) rectangle (15,3);
			\draw [arrows={-stealth}] (1,2.83)--(2,2.83); \draw [arrows={-stealth}] (2, 3.166)--(1,3.166); \draw [arrows={-stealth}] (4,2.83)--(3,2.83); \draw [arrows={-stealth}] (3, 3.166)--(4,3.166);
			\draw [arrows={-stealth}] (1,0.5)--(4,0.5); \draw [arrows={-stealth}] (4,0.83)--(1,0.83); \draw [arrows={-stealth}] (1,1.16)--(4,1.16);
			\draw [arrows={-stealth}] (-1,1.33)--(0,1.33); \draw [arrows={-stealth}] (0,2.66)--(-1,2.66); 
			\draw [arrows={-stealth}] (10,1.66)--(11,1.66); \draw [arrows={-stealth}] (11,2.32)--(10,2.32); \draw [arrows={-stealth}] (14,1.66)--(13,1.66); \draw [arrows={-stealth}] (13,2.32)--(14,2.32);
			\draw [arrows={-stealth}, red] (15,2.665)--(16,2.665); \draw [arrows={-stealth}, red]  (16,2.332)--(15,2.332); \draw [arrows={-stealth}, red] (15,1.999)--(16,1.999); 
			\draw [arrows={-stealth}, red] (16,1.666)--(15,1.666); \draw [arrows={-stealth}, red]  (15,1.333)--(16,1.333);
			\node [align=center] at (0.5,4.5) {$\Pi_A$};
			\node [align=center, red] at (14.5,3.5) {$\sigma_B$};
			\node [align=center] at (8,2) {\large{$\approx_{\epsilon}$}};
		\end{tikzpicture}
		
		$$ \red{\exists\, \sigma_B \in \mathbbm S},  \quad d^{\mathbbm D}( \Pi_A \, R_A,   S_B\, \red{\sigma_B})\leq\eps$$
		
		\caption{When Bob is dishonest and Alice is honest, the resulting real system obtained by removing Bob's honest protocol must be $\eps$-simulatable by connecting a converter $\sigma_B$ to Bob's interface of the corresponding ideal system, $S_B$.}
	\end{subfigure}
	\caption{\label{fig: security}The three conditions from Definition~\ref{definition: security}.}
\end{figure}


A \emph{possibility result} for a construction $\mathcal R \xrightarrow{\Pi} \mathcal S$ with parameters $(\eps, \mathbbm S, \mathbbm D)$ is a statement of the form:  there exists a protocol $\Pi=(\Pi_A,\Pi_B)$ that $\eps$-constructs $\mathcal S$ from $\mathcal R$, i.e.
\begin{align} \exists \, \Pi_A, \Pi_B, \sigma_A, \sigma_B \in \mathbbm S, 
\quad \forall \, \mathcal D \in \mathbbm D, \quad & d^{\mathcal D}(\Pi_A R \Pi_B, S)\leq \eps, \label{eq:security.0} \\
    & d^{\mathcal D}( R_A \Pi_B, \sigma_A S_A)\leq \eps, \label{eq:security.A} \\
    & d^{\mathcal D}( \Pi_A R_A, S_B \sigma_B)\leq\eps. \label{eq:security.B}
\end{align}
We then say that $\mathcal R$ is \emph{stronger} than $\mathcal S$.
An \emph{impossibility result} with the same parameters has the form:  there exists no  protocol $\Pi=(\Pi_A,\Pi_B)$ that $\eps$-constructs $\mathcal S$ from $\mathcal R$,
$$\forall \, \Pi_A, \Pi_B, \sigma_A, \sigma_B \in \mathbbm S, 
\quad \exists \, \mathcal D \in \mathbbm D, \quad \text{either condition \eqref{eq:security.0}, \eqref{eq:security.A}, or  \eqref{eq:security.B} does not hold}. $$

The strength of a security proof depends on the range of the class $\mathbbm S$ of simulators and protocols, the class  $\mathbbm D$ of distinguishers used in the security definition, as well as the assumed and constructed resources $\mathcal R$ and $\mathcal S$.
For construction results, a strong statement has the form  ``we can easily construct $\mathcal S$ from $\mathcal R$, and we can easily simulate any cheating behaviour, such that even a very powerful distinguisher could not tell apart our construction from the ideal system.'' Therefore, ideally we would want $\mathbbm S$ to be restricted to converters that are easy to implement physically, and we want the set of distinguishers $\mathbbm D$ to be as general as possible.
For impossibility results, a strong statement has the form ``we can always easily distinguish any system constructed from $\mathcal R$ from the resource $\mathcal S$, even if we allow for very powerful protocols and simulators.'' Therefore, we try to make $\mathbbm S$ to be as general as possible, and we restrict $\mathbbm D$ to correspond to efficient or otherwise easy to implement distinguishers.\footnote{In some settings, we may want to give more power to one of the players. This is the case for blind computation results \cite{Broadbent2009,Dunjko2014,Dunjko2016}, where for example Bob represents a client with limited computational power and Alice a powerful server (which may for example perform arbitrary quantum operations). In other examples, we may want to restrict honest players to use efficient protocols, while allowing the simulators of dishonest behaviour to be arbitrary. 
    In these and other cases, we can adjust the sets for $\Pi_A, \Pi_B, \sigma_A, \sigma_B $ and $\mathcal D$ to suit the scenario. For the results in this paper, this will not be necessary.}
    
We do not specify what $\mathbbm S$ and $\mathbbm D$ should be used in Definition~\ref{definition: security}, since this will be different for different theorems. For example, when we prove that no protocol can construct a biased coin flipping resource in Theorem~\ref{claim: impossibility}, the proof holds for converters $\Pi_A, \Pi_B, \sigma_A, \sigma_B \in \mathbbm S$ that have unbounded memory, unbounded computational power, and are post-quantum---they are only restricted to be non-signalling. The distinguisher $\mathcal D$ that is used to distinguish the real from ideal system runs these converters internally, and thus has the same computational and memory requirements as these converters.
    
\begin{remark}[Capturing bounded systems]
\label{remark: bounded systems}
Note that when a statement we want to prove involves an existence quantifier (over the set of converters $\mathbbm S$ for a possibility result, and over the set of distinguishers $\mathbbm D$ for an impossibility proof), it is not necessary to define the entire set ($\mathbbm S$, $\mathbbm D$), it is sufficient to convince oneself that the corresponding system does belong in this set. We use this to prove impossibility results for computationally bounded adversaries as well as in the bounded and noisy storage models in Sec.~\ref{sec:results} without defining either the complexity of the systems or the bound on the storage. We achieve this by finding a distinguisher that can distinguish real from ideal systems, and does so by internally running instances of these systems. This means that security already breaks down when the rest of the world (captured by the distinguisher $\mathcal D$) has the same memory bounds as the honest players and simulator in the protocol. Since a model needs the distinguisher to have at least the same power as the players and simulator for a protocol to be composable with itself, our impossibility results holds for any such model, regardless of the exact bounds on the computational power or storage, and irrespective of how this is defined.
\end{remark}

\subsection{Cryptography in relativistic settings:  the Causal Boxes framework \texorpdfstring{\cite{Portmann2017}}{CBcite}}
\label{ssec: CB}

The abstract cryptography framework \cite{Maurer2011} follows a top-down approach to modelling cryptographic security starting from the highest level of abstraction and proceeding downwards, introducing at each level only the minimum necessary specifications. The composability of abstract systems in the abstract cryptography framework makes it possible to provide a general, composable security definition, which is independent from the models of communication or computation. It can then be instantiated with whatever model is needed---here, Causal Boxes to model relativistic cryptography. In this section we give a brief, informal overview of the Causal Boxes framework. A formal introduction may be found in Appendix~\ref{appendix:causalboxes}.

Causal boxes \cite{Portmann2017} are a model of information-processing systems which may interact with each other in arbitrary ways, so long as they respect causality (Fig.~\ref{fig: CB}). 
In broad lines, a causal box $\hat{\Phi}$ is a system  with  input and output wires which may carry quantum or classical information. A concrete example is a physical box containing some optical elements (like beam-splitters) and  connected to optical fiber cables: each wire may carry several messages at different times (or even in a superposition of different times). 
A single instance of a message is modelled as a quantum state in the joint Hilbert space $\mathcal H \otimes l^2(\mathcal{T})$, where $\mathcal H$ is the Hilbert space of the actual classical/quantum message, $\mathcal T$ is a partially ordered set that defines an ordering on the space of messages and $l^2(\mathcal{T})$ is the sequence space with bounded 2-norm.\footnote{This is the state space of a single input/output message. More generally, wires which can carry messages in a superposition of different numbers and time orderings can be represented by the symmetric Fock space of this message space \cite{Portmann2017}. The symmetry comes from the fact that there is no special ordering of the messages other than the space-time ordering, which is already given in the state description itself. See Appendix~\ref{appendix:causalboxes} for further details.}  
In the simple cases where a quantum state $\rho \in \mathcal{H}$ is sent at a well-defined space-time coordinate $P \in \mathcal{T}$, $\mathcal{T}$ can be taken to be Minkowski space-time and we can simply represent the total state as a pair $(\rho, P)$. In this paper we only need to consider such cases.

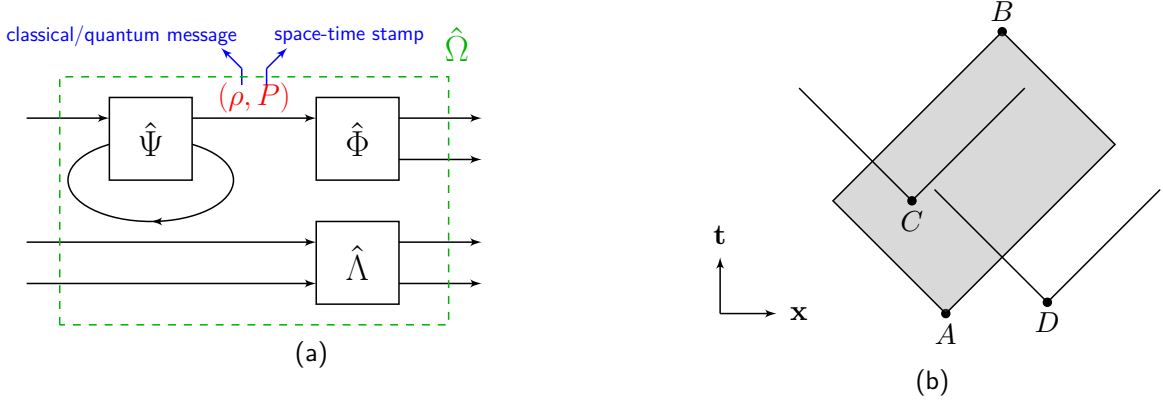
\begin{figure}[t]
\centering
\begin{subfigure}{0.5\textwidth}
	\begin{tikzpicture}[line width=0.20mm, scale=0.55, transform shape]
		\draw[decoration={markings, mark=at position 0.75 with {\arrowreversed{latex'}}}, postaction={decorate}](1,0) ellipse (2 and 1);
    	\draw[fill=white] (0,0) rectangle node{\Huge{$\hat{\Psi}$}}(2,2); \draw (5,0) rectangle node{\Huge{$\hat{\Phi}$}}(7,2);
		\draw [arrows={-latex'}] (-2,1.5)--(0,1.5); \draw [arrows={-latex'}] (2,1.5)--(5,1.5); \node[align=center, red] at (3.5,2) {\huge{\textbf{$(\rho, P)$}}};
		\draw [blue] (3.2,2.3)--(3.2,2.7); \draw [arrows={-latex'}, blue] (3.2,2.7)--(2.7,3.2); \node [align=center, blue] at (0.3,3.5) {\Large{classical/quantum message}};
		\draw [blue] (3.8,2.3)--(3.8,2.7); \draw [arrows={-latex'}, blue] (3.8,2.7)--(4.3,3.2); \node [align=center, blue] at (5.8,3.5) {\Large{space-time stamp}};
		\draw [arrows={-latex'}] (7,0.5)--(9,0.5);  \draw [arrows={-latex'}] (7,1.5)--(9,1.5);
		\draw (5,-3) rectangle node{\Huge{$\hat{\Lambda}$}}(7,-1); \draw [arrows={-latex'}] (-2,-1.5)--(5,-1.5); \draw [arrows={-latex'}] (-2,-2.5)--(5,-2.5);
		\draw [arrows={-latex'}] (7,-1.5)--(9,-1.5); \draw [arrows={-latex'}] (7,-2.5)--(9,-2.5);
		\draw[dashed, black!30!green] (-1.2,-3.5) rectangle (8.2,2.5); \node[align=center, black!30!green] at (8.4,3.3) {\Huge{$\hat{\Omega}$}};
	\end{tikzpicture}
	\caption{ }
	\label{fig: CB}
\end{subfigure}%
\begin{subfigure}{0.5\textwidth}
\centering
\begin{tikzpicture}[line width=0.20mm, scale=1.5]
    \filldraw[black] (0,0) circle (1pt) node[anchor=north] {$A$}; \filldraw[black] (0.5,2.5) circle (1pt) node[anchor=south] {$B$};
    \draw[fill=gray, fill opacity=0.3] (0.5,2.5)--(-1,1)--(0,0)--(1.5,1.5)--cycle; 
    \filldraw[black] (-0.3,1) circle (1pt) node[anchor=north] {$C$}; \filldraw[black] (0.9,0.1) circle (1pt) node[anchor=north] {$D$};
    \draw (-0.3,1)--(0.7,2); \draw (-0.3,1)--(-1.3,2); \draw (0.9,0.1)--(1.9,1.1); \draw (0.9,0.1)--(-0.1,1.1);
    \draw[arrows={-latex'}] (-2,0)--(-2,0.5); \draw[arrows={-latex'}] (-2,0)--(-1.5,0); \node[align=center] at (-2,0.7) {$\vec{t}$}; \node[align=center] at (-1.3,0) {$\vec{x}$};
\end{tikzpicture}
\caption{}
\label{fig: causaldiamond}
\end{subfigure}
\caption{{ \bf a. Causal boxes} are information-processing systems that respect causality and are closed under composition (serial, parallel or loops). Arbitrary composition of the causal boxes $\hat{\Phi}$, $\hat{\Psi}$ and $\hat{\Lambda}$ is a causal box $\hat{\Omega}$.  
{\bf b. Minkowski space-time.} The causal diamond of the space-time points $A$ and $B$ (shaded in gray) with $A\prec B$ is denoted by $D(A,B)$. In this figure, point $C\in D(A,B)$, and point $D$ is space-like separated from $A$ since the future light cone of neither of the points completely contains the future light cone of the other.}
\label{fig: causality}
\end{figure}

\paragraph{Causality condition.}
Causality requires that outputs produced at space-time point $P\in \mathcal{T}$ can depend only on inputs produced in its causal past, $P' \prec P$ (at this stage, $\cT$ could be any set of points equipped with any partial order to represent causality). In general, a causal box is a map from the space of the inputs to the space of the outputs that respects this notion of causality.\footnote{Technically, this implies that there must necessarily be a finite time gap between an input to a causal box and an output that depends on this input modelling the fact that any causal information processing task takes a strictly non-zero amount of time.}
Composition of causal boxes may be done in series,  in parallel or through (feedback) loops (Fig.~\ref{fig: CB}), and arbitrary composition of causal boxes results in a causal box. 
A more technical and detailed description of the framework can be found in Appendix~\ref{appendix:causalboxes}.

\paragraph{Minkowski space-time.} 
In this paper we apply the formalism of Causal Boxes to Minkowski space-time $\mathcal T$, where each coordinate corresponds to a vector $P= (\vec x,t)$ with three dimensions of space and one of time. 
In special relativity, $\mathcal T$ has a natural partial order, ``$P_1=(\vec x_1,t_1) \prec P_2=(\vec{x}_2, t_2)$ if light can reach $\vec {x}_2$  from $\vec x_1$ in time $t_2-t_1$, that is 
if $\| \vec x_2 - \vec x_1 \| \leq c(t_2 - t_1)$, where $c$ is the speed of light.'' In this case we say that space-time point  $P_1$ is in the causal past of $P_2$. If two points are not ordered, we say that they are space-like separated.
The causal diamond of a pair of space-time points, $P_1\prec P_2$, denoted by $D(P_1,P_2)$ is the intersection of the future light cone of $P_1$ with the past light cone of $P_2$. This represents the maximal space-time region that can be affected by events at $P_1$ and also affect events at $P_2$ (Fig.~\ref{fig: causaldiamond}).
In the following, we assume that all players involved in a relativistic cryptographic protocol initially agree upon a coordinate system to represent all space-time points.

\begin{remark}[Range of causal boxes]
\label{remark: cb_range}
Causal Boxes can model not only quantum processes, but also non-signalling systems with quantum and classical inputs (for example, PR-boxes are causal boxes) \cite{Portmann2017}. 
This will be useful in security proofs, for example to cover very powerful adversaries, so let us denote by $\mathbbm C$ the set of all allowed causal boxes in $\cT$, and by $\mathbbm{D}_{\mathbbm C} \subset \mathbbm C$ the subset of systems that are valid distinguishers.

When proving the possibility result in Section~\ref{ssec: cons}  (Theorem~\ref{claim: constructibility}), we show that
$$ \exists\, \Pi,\sigma \in \mathbbm S,  \quad \forall \, \mathcal D \in \mathbbm{D}_{\mathbbm C},\quad  d^{\mathcal D}( \Pi \, R,   S\, \sigma)\leq\eps,$$ where $\mathbbm S$ are just efficient classical systems.
This means that even distinguishers bounded only by non-signalling constraints cannot distinguisher the real from ideal systems, and the construction still holds in the presence of such unrestricted adversaries.

When proving impossibility results in Sections~\ref{ssec: impos} and \ref{ssec: impos_delay}, we show that
$$ \forall\, \Pi,\sigma \in \mathbbm S,  \quad \exists \, \mathcal D \in \mathbbm{D}_{\mathbbm S},\quad  d^{\mathcal D}( \Pi \, R,   S\, \sigma)>\eps,$$ where $\mathbbm S \subset \mathbbm C$ is any set of systems (e.g., classical, computationally limited or with bounded memory) and $\mathbbm{D}_{\mathbbm S}$ is a set of distinguishers with similar requirements. This means firstly that our impossibility results hold even if we consider protocols that are bounded only by non-signalling constraints (the case were $\mathbbm S = \mathbbm C$). And secondly, if we consider a setting where adversaries are limited, then the results carry over to this setting. For example, our impossibility proofs also hold in the bounded storage model (where $\mathbbm S$ and $\mathbbm D_{\mathbbm S}$ have bounded memory) or a computational setting (where $\mathbbm S$ and $\mathbbm D_{\mathbbm S}$ are computationally limited). See also Remark~\ref{remark: bounded systems} in Sec.~\ref{ssec: AC}.
\end{remark}

\subsection{Two-party resources}
\label{ssec: resources}

We may now define the resources needed to model and prove our results. In this section, we model these resources by defining their output values and space-time positions given input values and space-time positions. As in illustration of how this is a special case of the more complete Causal Box model instantiated with Mikowski space, we provide in Appendix~\ref{ssec: appendix.CD} a formal definition of a \CD\ as a causal box.

\subsubsection{Coin flipping (\CF)}
\label{sssec: CF}

A coin flip resource provides two distrustful players with a random coin flip---if they both behave honestly. If one of them is dishonest, then the literature defines different resources that could be constructed. The most common, e.g., \cite{Blum1983}, is to allow the coin flip to be \emph{unfair}: a dishonest player who does not like the outcome can abort before the honest player gets to see this outcome. In \cite{Demay2013}, the authors define a \emph{biased} coin flip, where instead of aborting, a dishonest party can bias the outcome. In this section we follow \cite{Demay2013} and define a $p$-biased coin flip $\CF^p$. We define an unfair coin flip $\CF^{\text{uf}}$ in Appendix~\ref{appendix:unfairCF}, where we prove that $\CF^{1/2}$ can be constructed from $\CF^{\text{uf}}$.

\begin{definition}[Coin flipping, $\CF^p$]
\label{definition: CF}

A $p$-biased coin flip, $\CF^p = \{CF, CF^p_A, CF^p_B\}$, is defined as follows.
\begin{multicols}{2}
\begin{description}
    \item[ $CF$:] Alice receives a uniformly random bit $c$ at location $P$, and Bob receives the same bit at location $P'$.
    
    \item[$CF^p_B$:] Dishonest bob receives his coin flip output $c$ in advance at location $P_1$ and at location $P_2\succ P_1$ he may input a bit $b$ (which may depend on the value of $c$). Alice receives a bit $c_o^A$ at location $P\succ P_2$: with probability $p$ she receives $c_o^A = b$, else  $c_o^A = c$.   
    Causality requirement: $P_1\prec P_2 \prec P$. 

    \item [$CF^p_A$:] analogous to $CF^p_B$, with the roles reversed.
\end{description}

\begin{center}
		\begin{tikzpicture}[line width=0.20mm, scale=0.8]
			\draw (0,0) rectangle node{$CF$} (1.5,1.5); \draw [arrows={-stealth}] (0,0.75)-- node[midway,above]{\small{($c$, $P$)}}(-1.5,0.75); 
			\draw [arrows={-stealth}] (1.5,0.75)-- node[midway,above]{\small{($c$, $P'$)}}(3,0.75);
			\node [align=center] at (-2.7,0.75) {\textbf{Alice}};
			\node [align=center] at (4,0.75) {\textbf{Bob}};
		\end{tikzpicture}
		
\vspace{16pt}

		\begin{tikzpicture}[line width=0.20mm, scale=0.8]
			\draw (0,0) rectangle node {$CF^p_B$} (1.5,1.5); \draw [arrows={-stealth}] (0,0.75)--(-1.5,0.75); 
			\node [align=center] at (-1.55,1.05) {\small{($c_o\in \{c,b\}$, $P$)}};
			\draw [arrows={-stealth}] (1.5,1.2)-- node[midway,above]{\small{($c$, $P_1$)}}(3,1.2); \draw [arrows={-stealth}] (3,0.3)--(1.5,0.3); \node[align=center] at (3.1,0.6) {\small{($b \in \{ 0,1\}$, $P_2$)}};
			\node [align=center, red] at (0.75,-0.75) {$P_1\prec P_2\prec P$};
		\end{tikzpicture}

\end{center}
\end{multicols}
\end{definition}

Note that by definition of $CF$, it should be clear that the uniformly random bit, $c$ is generated independently by the resource $CF$ and cannot be correlated with anything outside it. This is because the honest resource $CF$ takes no inputs that could possibly influence this output. Further, a bias of $0$ means that the coin flip is uniform, a bias of $1$ means that the dishonest player has complete control over the outcome, and a bias of $p$ means that any outcome can occur with probability at most $1/2+p/2$. 

\subsubsection{Bit commitment (\BC)}
\label{sssec: BC}

As mentioned in the introduction, bit commitment is an important cryptographic primitive and its security relates to its \emph{hiding} and \emph{binding} properties which were also introduced in Section~\ref{sec:introduction}. Here, we formally define what an ideal bit commitment resource behaves like in Minkowski space-time.

\begin{definition}[Bit commitment, \BC]
\label{definition: BC}
A bit commitment resource tuple $\mathcal{BC}:=\{BC,BC_A,BC_B\}$ is defined by the single resource $BC$ (with $BC_A$ and $BC_B$ identical to $BC$), which behaves as follows.
\begin{multicols}{2}
\begin{enumerate}
   \item Alice selects a classical bit $a\in \{0,1\}$ to commit to and inputs it at her interface of $BC$ at a time of her choice $t_1$. 
   \item Bob receives the message `$comm$' at time $t'_1 > t_1$ at his interface, indicating that Alice has committed to a bit. 
   \item Alice then inputs the command `$open$' at her interface at a time of her choice $t_2$.
   \item Her original commitment `$a$' is then revealed to Bob at time $t'_2 > t_2$.
\end{enumerate}
\begin{center}
	\begin{tikzpicture}[line width=0.20mm, scale=0.7]
		\draw (0,0) rectangle node{\BC} (2,2); \draw [arrows={-stealth}] (-1.5,1.32)-- node [midway,above] {\small{$(a, t_1)$}}(0,1.32); \draw [arrows={-stealth}] (-1.5,0.66)--(0,0.66);\node[align=center] at (-1,1) {\small{$(open, t_2)$}};
		\draw [arrows={-stealth}] (2,1.32)--(3.5,1.32); \node[align=center] at (3.5,1.7) {\small{$(comm, t'_1)$}}; \draw [arrows={-stealth}] (2,0.66)--(3.5,0.66);\node[align=center] at (3,1) {\small{$(a, t'_2)$}};
		\node [align=center] at (-3.5,1.1) {Alice};
		\node [align=center] at (5.5,1.1) {Bob};
	\end{tikzpicture}
\end{center}
\end{multicols}
\end{definition}

For simplicity, we only mention the times at which the messages are input and output in Definition~\ref{definition: BC}. This should naturally also include the location in space of the players.

\subsubsection{Channel with delay (\CD)}
\label{sssec: CD}

In special relativity, unless two agents meet at the exact same space-time location to exchange messages, there is necessarily a finite communication delay between them. A \emph{channel with delay} is a cryptographic primitive between two parties based on this physical intuition: Alice sends a message and Bob receives it unaltered with some delay.

\begin{definition}[Channel with delay]
\label{definition: CD}    
A channel with delay $\CD = (CD, CD_A, CD_B)$ between a sender Alice and a receiver Bob is a tuple of resources characterized by four space-time locations, $P \prec  P'\prec Q' \prec Q$, and defined as follows.
\begin{multicols}{2}
\begin{description}
    \item{$CD$:}
    Honest Alice inputs a quantum state $a$ into the channel at location $P$, i.e., the input message is $(a, P)$. Honest Bob receives $(a, Q)$ at  location $Q$.
    
    \item{$CD_A$:} Dishonest Alice inputs $(a, P')$. Honest Bob receives $(a, Q)$.
    
    \item{$CD_B$:} Honest Alice inputs $(a, P)$. Dishonest Bob receives  $(a, Q')$. 
\end{description}
\noindent
The \emph{\safezone}\ of the channel is defined as the causal diamond of $P'$ and $Q'$: the set $D(P',Q') := \{T: P' \prec T \prec Q'\}$. 
\begin{center}
	\begin{tikzpicture}[line width=0.20mm, scale=1]
		\draw[arrows={-latex'}] (-2,-0.5)--(-2,0); \draw[arrows={-latex'}] (-2,-0.5)--(-1.5,-0.5); \node[align=center] at (-2,0.2) {$\vec{t}$}; \node[align=center] at (-1.3,-0.5) {$\vec{x}$}; 
		\filldraw[black] (0,0) circle (0.5pt) node[anchor=north] {$P$}; \filldraw[black] (0.5,2.5) circle (0.5pt) node[anchor=south] {$Q$};
		\node[align=center] at (0.8,0) {\textbf{ALICE}}; \node[align=center] at (-.2,2.5) {\textbf{BOB}};
		\filldraw[black] (-0.1,0.7) circle (0.5pt) node[anchor=north] {$P'$}; \filldraw[black] (0.4,1.8) circle (0.5pt) node[anchor=south] {$Q'$};
		\draw (0.5,2.5)--(-1,1); \draw (0.5,2.5)--(1.5,1.5); \draw(0,0)--(1.5,1.5); \draw (0,0)--(-1,1);
		\draw[blue] (-0.1, 0.7)--(0.7,1.5); \draw[blue] (-0.1,0.7)--(-0.4,1); \draw[blue] (0.4, 1.8)--(-0.4,1);  \draw[blue] (0.4, 1.8)--(0.7,1.5);
		\draw [draw=none, fill=blue, fill opacity=0.2] (-0.1,0.7)--(-0.4,1)--(0.4,1.8)--(0.7,1.5);
		\draw[blue, fill=blue, fill opacity=0.2] (-2,-1) rectangle (-1.8,-0.8);
		\node[align=center] at (-.5,-0.9) {trusted region};
	\end{tikzpicture}
\end{center}
\end{multicols}
\end{definition}

That is, the \CD\ acts as an identity channel on the message, and as a shift on the space-time stamp. Furthermore, it allows dishonest players to send (respectively, recieve) the message after (respectively, before) the honest player. A formal definition of the causal box that implements the \CD\ can be found in Appendix~\ref{ssec: appendix.CD}. 
The \safezone\ of the \CD\ is the region where both players can be sure that the information in the channel remains secure, even when the other is dishonest; as we will see, it is the region where the \CD\ can be used to construct other resources such as \CF\ (Section~\ref{ssec: cons}).

\paragraph{Relation to relativistic bit commitment protocols.}

Typically, in a non-relativistic bit commitment resource, Alice is free to choose when to open her commitment and also has the choice to not open her commitment at all. In relativistic protocols, however, the commitment time is usually restricted by the time taken by light to travel between the remote agents, in which case Alice does not have the freedom of choosing arbitrary $t_1$ and $t_2$ as in Definition~\ref{definition: BC}: once $t_1$ is fixed, the commitment must be opened at the latest by $t_1 + \Delta t$ for some $\Delta t$ which depends on the protocol. Bob typically does not know whether Alice is committed before time $t_1 + \Delta t$, when by checking if he has a valid commitment or garbage, he can know whether she ran the honest protocol at $t_1$ or not, and retroactively decide if she has been committed to her bit. Furthermore, in some relativistic protocols, e.g., \cite{Kent2012}, Alice cannot choose to not open: if she honestly committed at time $t_1$, then after $\Delta t$, the commitment is always opened. The \CD\ resource from Definition~\ref{definition: CD} captures exactly this, and hence we analyze the (im)possibility of extending the delay of such a channel in this work.

Other protocols, e.g., \cite{Lunghi2015}, additionally offer the possibility to Alice of aborting before Bob receives the bit to which she committed. We thus define a variation of Definition~\ref{definition: CD} in Appendix~\ref{appendix:CDabort}, where after inputting her message into the channel, Alice may still change her mind and abort before Bob receives it. We prove in Appendix~\ref{appendix:CDabort} that our main results presented in Sec.~\ref{sec:results} still go through with this alternative definition of a channel with delay.

%% file: results.tex
\subsection{Constructing \CF} 
\label{ssec: cons}

It was shown in \cite{Demay2013} that a $1/2$-biased coin flipping resource can be perfectly constructed from a bit commitment resource (Definition~\ref{definition: BC}), by using Blum's protocol \cite{Blum1983}. 
Here we show that it is in fact possible to construct an even stronger resource (an unbiased coin flip) from a channel with delay.

\begin{restatable}[Construction $\CD \to \CF$]{theorem}{ClaimConstructibility}
\label{claim: constructibility}
Given a classical channel with delay $\CD$, there exists a classical protocol $\Pi_{\mathcal{CD}\rightarrow \mathcal{CF}}=\{\Pi_A,\Pi_B\}$ that perfectly constructs an unbiased coin flipping resource $\CF^0$.

The constructed and ideal resources are indistinguishable for any possible distinguisher (including quantum and non-signalling distinguishers, see Remark~\ref{remark: cb_range} in Sec.~\ref{ssec: CB}). The honest protocol as well as the simulator require only elementary local operations and classical communication.
\end{restatable}

The protocol is described in Definition~\ref{definition: protocol}, and the security proof is given in  Appendix~\ref{appendix: cons}. 

\begin{definition}[Protocol $\Pi_{\CD \to \CF}$]
\label{definition: protocol}

Given a channel with delay  $\CD=(CD, CD_A, CD_B)$ characterized by locations $A\prec A'\prec B' \prec B$ (see Definition~\ref{definition: CD}), we define the following honest protocol $\Pi_{\CD \to \CF}= (\Pi_A, \Pi_B)$.

\begin{enumerate}
    \item  Alice picks a uniformly random bit, $a$ and sends it through $CD$ from her space-time location $A$. Bob receives this bit from $CD$ at his location $B$.
    \item Bob meets Alice at $P$ in the \safezone, i.e., the causal diamond $D(A',B')$ to pass on Bob's uniformly random bit, $b$.
    \item After receiving $b$ from her agent, Alice computes $a\oplus b=c$ and outputs this value at some point $P_F^A\succ P$. If Bob did not turn up for the meeting at $P$, she picks a uniform $b$ herself, and outputs $a\oplus b=c$ as before.
    \item After receiving $a$ from the channel, Bob computes $a\oplus b=c$ and outputs the result at a point $P_F^B\succ B$. If Bob does not receiving anything from the channel, he picks a uniform $a$ himself, and outputs $a\oplus b=c$ as before.
\end{enumerate}

\end{definition}

Note that it is important that the point $P$ in the above protocol lies in the \safezone, otherwise the protocol would not be secure.\footnote{The existence of the simulators $\sigma_A$ and $\sigma_B$ used in the proof of Theorem~\ref{claim: constructibility} relies crucially on $P \in D(A',B')$.} Furthermore, this protocol can be run by a single player on each side without the need for trusted agents since $P$ lies in the causal future of $A$ and $A'$ and in the causal past of $B$ and $B'$.

In Appendix~\ref{appendix:CDabort} we define a weaker channel with delay, namely one which allows Alice to abort and prevent her message from reaching Bob, $\CD^\bot$. We show in the same appendix (Lemma~\ref{lem:CDabort}), that if the protocol above is used with $\CD^\bot$ instead of $\CD$, then we construct an unfair coin flip $\CF^{\text{uf}}$ instead of an unbiased one $\CF^0$.

\subsection{Impossibility of \CF, \CD\ and \BC}
\label{ssec: impos}

\paragraph{Impossibility of coin flipping.} In the previous section, we showed that an unbiased coin flipping resource can be constructed from a suitable channel with delay. Here we show that in the absence of any such shared resource, it is impossible to construct any (biased) coin flip resource solely through the exchange of messages.

\begin{restatable}[Impossibility of \CF]{theorem}{ClaimImpossibilityCF}
\label{claim: impossibility}
It is impossible to construct, with $\epsilon<\frac{1}{6}(1-p)$, a $p$-biased coin flipping resource between two mutually distrusting parties solely through the exchange of messages through any relativistic or non-relativistic protocol, be it classical, quantum or non-signalling.

The distinguisher required to distinguish the real form ideal systems has the same complexity and memory requirements as the protocol $\Pi_A,\Pi_B$ and simulators $\sigma_A,\sigma_B$. In particular, if these are efficient, classical or have bounded or noisy memory, then so does the distinguisher.
\end{restatable}

Note that this theorem includes as special case protocols that may send messages in superpositions of different causal orders. This follows from the fact that the impossibility holds for any causal boxes, thus in particular for causal boxes that use such superpositions of causal orders.

The proof of Theorem~\ref{claim: impossibility} can be found in Appendix~\ref{appendix: impos}. Here below we provide some intuition.

A coin flip $\CF^p$ does not only guarantee that the output bit is uniform (or biased with probability $p$), but also that it is independent of any other bit produced in parallel by some other resource (up to the bias). This is essential so that a dispute that is resolved with a coin flip would not only be settled fairly, but also independently from any other dispute. The man in the middle attack mentioned in Sec.~\ref{sec:introduction} would allow dishonest players to perfectly correlate the outcome of two coin flips that are expected to be independent: if Alice and Bob run a coin flipping protocol, Charlie and Danielle run a second one in parallel, and Bob and Charlie collude to forward all the communication between Alice and Danielle, Bob and Charlie could force them to agree on the same coin flip. The proof of Theorem~\ref{claim: impossibility} consists in showing that this is essentially possible for any protocol that does not use any resource other than communication between the parties involved. A sketch of the main proof idea is provided in \fref{fig: impossibility} in Appendix~\ref{appendix: impos}. It generalizes the techniques used in \cite{Maurer2011} to prove the analogous result for the non-relativistic case.

\paragraph{Impossibility of \BC\ and \CD.} Combined with Theorem~\ref{claim: constructibility} and Blum's construction \cite{Demay2013}, Theorem~\ref{claim: impossibility} implies impossibility of constructing any channel with delay $\CD$ or any commitment $\BC$ of no initial resource is shared by the players.

\begin{restatable}[Impossibility of \CD]{corollary}{ClaimImpossibilityCD}
\label{claim: cdimpos}
It is impossible to construct \CD, with $\epsilon<\frac{1}{6}$, between two mutually distrusting parties solely through the exchange of messages through any classical, quantum or relativistic protocol.

The distinguisher required to distinguish the real form ideal systems has the same complexity and memory requirements as the distinguisher used in Theorem~\ref{claim: impossibility} composed with the protocol $\Pi_A,\Pi_B$ used in Theorem~\ref{claim: constructibility}. In particular, if these are efficient, classical and have bounded or noisy memory, then so does the distinguisher.
\end{restatable}

\begin{proof}
Follows directly from the impossibility of \CF\ in Theorem~\ref{claim: impossibility} together with the construction of unbiased \CF\ from \CD\ (Theorem~\ref{claim: constructibility}). 
\end{proof}

\begin{restatable}[Impossibility of \BC]{corollary}{ClaimImpossibilityBC}
\label{claim: bcimpos}
It is impossible to construct $\BC$, with $\epsilon<\frac{1}{12}$, between two mutually distrusting parties solely through the exchange of messages through any classical, quantum or relativistic protocol. This rules out both arbitrarily long and timed commitments. 

The distinguisher required to distinguish the real form ideal systems has the same complexity and memory requirements as the distinguisher used in Theorem~\ref{claim: impossibility} composed with the protocol $\Pi_A,\Pi_B$ used in Blum's protocol \cite{Demay2013}. In particular, if these are efficient, classical and have bounded or noisy memory, then so does the distinguisher.
\end{restatable}

\begin{proof}
Follows directly from the impossibility of \CF\ in Theorem~\ref{claim: impossibility} together with the construction of $\frac12$-biased \CF\ from \BC\ using Blum's protocol \cite{Demay2013}. 
\end{proof}

Using the same techniques, we show in Appendix~\ref{appendix:CDabort} that an abort channel cannot be constructed either.

\subsection{Impossibility of Extending Delays}
\label{ssec: impos_delay}

We show that it is not possible to use several channels with delay to construct a better channel with delay: the \safezone\ of the constructed channel will be smaller than the \safezone\ of at least one of the individual channels used. In fact, the result is even stronger: the \safezone\ of the constructed channel is contained inside the \safezone of at least one of the assumed channels used in the construction. This means the maximal space-time region within which the information in the channel is guaranteed to be secure from both dishonest parties cannot be increased even with $n$ copies of a channel. If we view such a channel with delay as a relativistic bit commitment (Alice inputs a bit into the channel and is then committed to it, but the commitment is only opened when the bit arrives with a delay at Bob), this implies that it is not possible to increase the time within which the commitment is both hiding and binding even if $n$ timed commitment resources are given.

\begin{restatable}[Impossibility of extending \CD]{theorem}{ClaimImpossibilityExtendingCD}
\label{claim: impos2}
Given $n$ channels with delay  $\mathcal{CD}^1$,...,$\mathcal{CD}^n$ between two parties, it is impossible to construct with $\epsilon \leq \frac{1}{8}$  a channel $\mathcal{CD}'$ between the two parties with a \safezone\ that is larger than the \safezone\ of all of the individual channels used.

This holds for all protocols $\Pi_A,\Pi_B$ in $\mathbbm C$, which includes inefficient and non-signalling systems. The distinguisher needed to distinguish the real from ideal system has the same complexity requirements as the protocol $\Pi_A,\Pi_B$. In particular, if it is efficient or classical, then so is the distinguisher. Furthermore, if the channels constructed and used are classical, then the distinguisher also has the same quantum memory requirements as the protocol $\Pi_A,\Pi_B$.
\end{restatable}

The proof of Theorem~\ref{claim: impos2} is given in Appendix~\ref{appendix: impos_delay}. Note that this proof includes as special case protocols and distinguishers that may send messages in a superposition of going through one channel $\CD^i$ or another $\CD^j$, or in which a channel may be in a superposition of being used and not used. This follows from the fact that the impossibility holds for any causal boxes, thus in particular for causal boxes that use such superpositions of causal orders.

One may consider variations of this result in which slightly different resources are used or constructed. For example, one could wonder whether having channels with delay going from Bob to Alice may help. It is however easy to verify from that proof, that these have no impact on the impossibility. Another variation worth considering is if the channels are abort channels, as defined in Appendix~\ref{appendix:CDabort}. We prove in the same appendix that one cannot extend the delay of abort channels either.

%% file: discussion.tex
The general framework for modelling composable security of relativistic quantum protocols developed here naturally lends itself to the study of novel possibility and impossibility results in relativistic cryptography and could provide key insights into classifying possible and impossible information-processing tasks.  

\paragraph{Composability issues raised previously.} Composability issues with Kent's 2012 protocol \cite{Kent2012} have been briefly discussed in \cite{Kaniewski2013}. A definition which is labeled ``composable'' is proposed in \cite[Appendix B]{Kaniewski2013}, but it is not derived using any composable framework. In fact, it is argued in \cite{Kaniewski2013} that bit commitment in the bounded and noisy storage models could satisfy this definition. Since our results carry over to these settings as well, it follows that either the proposed definition is not composable or it cannot be satisfied. Note that the impossibility of $\BC$ in the bounded storage model is already hinted at in \cite{Unr11}, where the author points out that the model he developed for concurrent composition does not guarantee that a protocol is secure when run in parallel with another instance of itself.

\paragraph{Superpositions of causal orders.} A unique feature of the the causal boxes formalism \cite{Portmann2017}, is that it can model superpositions of messages exchanged in a superposition of orders in (space-)time (e.g.\ the quantum switch \cite{Chiribella2013}) by assigning different space-time stamps (or superpositions thereof) to different messages. Combining this with the abstract cryptography framework \cite{Maurer2011}, as done in this paper, allows us to model security in settins where such superpositions are actively used. For example, this allows us to consider protocols where a message is in a superposition of being sent from Alice to Bob and from Alice to Charlie, i.e., where Bob and Charlie are in a superposition of having received no message and one message from Alice. Even for protocols that do not use such structures, possibility results consider distinguishers that have this capability. And impossibility results show that even such such superpositions of causal orders, the desired resource cannot be constructed. This is the case for all our results presented in this work.

A known example of a process involving a superposition of temporal orders of operations is the quantum switch \cite{Chiribella2013}. It was physically realized in \cite{Procopio2015,Rubino2017}, and  can be represented in the Causal Box framework as shown in \cite{Portmann2017}. Further, the quantum switch was shown to have an operational advantage over fixed ordering of (or classical mixtures thereof) operations in solving certain computational tasks \cite{Colnaghi2011,Araujo2014}. By modelling cryptographic protocols involving such superpositions of orders, one can study the operational advantage provided by quantum ordering of messages/operations over classical orderings. Such an approach to studying causal structures in terms of their operational advantages would be useful for characterising the properties of physically implementable causal structures. This is still an important open question since there exist more general frameworks for modelling causal structures, such as the \emph{process matrix framework} \cite{Oreshkov2012} which predict causal structures that are logically possible and yet, have no known physical implementation.

\paragraph{Error tolerance.} Realistic protocols, like those implemented with quantum preparations and measurements, always come with a small probability of error (for example, in Kent's protocol as in QKD schemes, this depends on the number of quantum states exchanged between the parties). The ideal resources we prove cannot be constructed are, by definition, not subject to any errors. But it follows directly from the composable framework used that impossibility to construct perfect resources (with some error $\eps$) implies impossibility to construct noisy versions of the resources. To see this, consider a resource $\CD^\eps$  that is $\eps$-close to \CD\ according to the distinguishing advantage. By the triangle inequality, if a real protocol implements a resource that is $\Delta$-distinguishable from the ideal \CD, it will be at least $(\Delta-\eps)$-distinguishable from $\CD^\eps$. 
For example, for an unbiased \CF, we have $\Delta=\frac16$, so it is still impossible to perfectly build any \CF\ that has an error tolerance smaller than that.

\paragraph{Minimal resources for constructions.}
Our results show that existing bit commitment protocols \cite{Kent2012,Lunghi2015} cannot construct the target resource \BC\  from an assumption of a shared resource. 
Nevertheless, we may still look for initial resources $\mathcal R$ that allow \BC\ to be constructed.
It would be interesting to explore the minimal resources necessary to achieve this. For example, an assumption (or assurance) that dishonest players cannot interact with third parties is a good candidate for such an initial resource $\mathcal R$. It remains open to formalize such a resource $\mathcal R$ within the framework and prove that it is sufficient to construct \BC.

\paragraph{Alternative space-time.}
We have proved our results taking the background physical theory to be special relativity (in the sense of Minkowski space-time with a finite speed of signalling). The results would still hold even for other space-time geometries with a \emph{fixed} background causal structure i.e., for different choices of the partially ordered set $\mathcal{T}$. However, if we consider a general relativistic framework (one where the causal order is not fixed until one solves for the metric by considering the mass distribution) that is compatible with quantum mechanics, there could arise situations where the background causal structure itself is subject to quantum uncertainty and is no longer fixed.\footnote{This can arise when large masses are superposed, resulting in the space-time geometry and hence the causal structure being in a superposition \cite{Zych2017}.} Such causal structures can no longer be explained by a single partially ordered set $\mathcal{T}$ and cannot be modelled within the Causal Boxes framework. In fact, there is currently no framework that can model this and has the properties required to define cryptographic security.\footnote{While Process Matrices \cite{Oreshkov2012} and Causaloids \cite{Hardy2005} are examples of frameworks that are capable of modeling such causal structures, they do not provide a model of discrete systems that can be composed, which is needed for cryptography.} Thus it remains open to define a quantum, general relativistic framework for composable cryptography, and study the problem of constructing bit commitment using it. 

%% file: appendix_causal_boxes.tex
The causal box \cite{Portmann2017} formalism models information-processing systems that are closed under composition even when the order of operations indefinite (such as a superposition of orders) or dynamically determined during a protocol's runtime. Similar formalisms (e.g., \cite{Chiribella2009}, \cite{Gutoski2010} and \cite{Hardy2012}) have been previously developed but they are only suitable for modelling systems where the order of messages is predefined, they fail to be closed under composition when considering simple cryptographic protocols that involve dynamical ordering of messages during runtime \cite{Portmann2017}. In particular, the formalism allows us to model quantum causal systems in Minkowski space and construct new causal systems by composing them. This makes it suitable for modelling composable security of relativistic quantum protocols as done in this paper. We now review the formal definitions of the objects of the causal box framework \cite{Portmann2017}.

\subsection{Message space and wires}
\label{sec: messagewire}

\begin{enumerate}
  \item \textbf{Space of ordered messages:} Every message is modelled as a pair, $(v, t)$ where $v\in \mathcal{V}$ denotes the (classical/quantum) message and $t\in \mathcal{T}$ provides ordering information, where $\mathcal{T}$ is a countable, partially ordered set. The space of a single message is thus a Hilbert space with the orthonormal basis $\{\ket{v,t}\}_{v\in \mathcal{V},t\in \mathcal{T}}$. For a finite $\mathcal{V}$ and infinite $\mathcal{T}$, this Hilbert space corresponds to $\mathbb{C}^{|\mathcal{V}|}\otimes l^2(\mathcal{T})$ where $l^2(\mathcal{T})$ is the sequence space with a bounded 2-norm. Thus $\ket{t}$ can be seen as a sequence which consists of a 1 in position $t \in \mathcal{T}$ and 0 everywhere else.
  \item \textbf{Wires:} The inputs and outputs to a causal box are sent/received through wires which can carry any number (or a superposition of different numbers) of messages of a fixed dimension, which defines the dimension of the wire\footnote{For example a $2$ dimensional wire can carry any number of qubits, or can be in a superposition of carrying 2 and 3 qubits but cannot carry qutrits.}. Thus the state space of a wire is defined to be a symmetric Fock space. It is modelled as a Fock space to allow for superpositions of different numbers of messages and it is a symmetric Fock space since all ordering information associated with the arriving qudits is already contained in the label $t\in \mathcal{T}$ and given this label, there is no other ordering on the qudits. For the Hilbert space, $\mathcal{H}=\mathbb{C}^d\otimes l^2(\mathcal{T})$, the corresponding bosonic Fock space is given as
  \begin{equation}
    \mathcal{F}(\mathbb{C}^d\otimes l^2(\mathcal{T})):= \bigoplus\limits_{n=0}^{\infty} \vee^n(\mathbb{C}^d\otimes l^2(\mathcal{T})),
  \end{equation}
  where $\vee^n\mathcal{H}$ denotes the symmetric subspace of $\mathcal{H}^{\otimes n}$ and $\mathcal{H}^{\otimes 0}$ is the one-dimensional space containing the vacuum state $\ket{\Omega}$.\\
\end{enumerate}

For example,  the state space corresponding to a wire $A$ carrying $d_A$ dimensional messages is denoted by $\mathcal{F}_A^{\mathcal{T}}=\mathcal{F}(\mathbb{C}^{d_A}\otimes l^2(\mathcal{T}))$. The joint space of two wires can be written as $\mathcal{F}_A^{\mathcal{T}}\otimes \mathcal{F}_B^{\mathcal{T}}=\mathcal{F}_{AB}^{\mathcal{T}}$ and it can be shown \cite{Portmann2017} that for any two Hilbert spaces $\mathcal{H}_A=\mathbb{C}^{d_A}\otimes l^2(\mathcal{T})$ and $\mathcal{H}_B=\mathbb{C}^{d_B}\otimes l^2(\mathcal{T})$, 
  \begin{equation}
  \label{eq: isomorphism}
  \mathcal{F}(\mathcal{H}_A)\otimes \mathcal{F}(\mathcal{H}_B) \cong \mathcal{F}(\mathcal{H}_A\oplus \mathcal{H}_B),
  \end{equation}
  
  Isomorphism~\ref{eq: isomorphism} tells us that each valid state in the combined state space of two wires, one carrying $d_A$ dimensional messages and the other carrying $d_B$ dimensional messages, can be mapped to a valid state in the state space of a single wire carrying $d_A+d_B$ dimensional messages. Hence $\mathcal{F}_{AB}^{\mathcal{T}}$, can be interpreted as the state space of a wire carrying $(d_A+d_B)$ dimensional messages\footnote{Conversely, any wire $A$ of messages of dimension $d_A$ can be split in two sub-wires $A_1$ and $A_2$ of messages of dimensions $d_{A_1} + d_{A_2} = d_A$: $\mathcal{F}_A^{\mathcal{T}}\cong \mathcal{F}_{A_1}^{\mathcal{T}} \otimes \mathcal{F}_{A_2}^{\mathcal{T}}$.
  Further, for any subset $\mathcal{P} \subseteq \mathcal{T}$, $\mathcal{F}_A^{\mathcal{T}}\cong \mathcal{F}_A^{\mathcal{P}}\otimes \mathcal{F}_A^{\widetilde{\mathcal{P}}}$,
  where $\widetilde{\mathcal{P}}=\mathcal{T} \backslash \mathcal{P}$ and $\mathcal{F}_A^{\mathcal{P}}=\mathbb{C}^{d_A}\otimes l^2(\mathcal{P})$.}.

\par
We now proceed to formally review the definition of causality that causal boxes satisfy, we first define the notion of cuts on a partially ordered set $\mathcal{T}$ which forms an important part of the definition.

\subsection{Cuts and causality}
\label{sec: cuts}
\begin{definition}[Cuts \cite{Portmann2017}]
A cut of a partially ordered set $\mathcal{T}$ is any subset $\mathcal{C}\subseteq \mathcal{T}$ such that 
$\mathcal{C}=\bigcup\limits_{t\in \mathcal{C}}\mathcal{T}^{\leq t}$,
where $\mathcal{T} \leq t = \{p \in \mathcal{T} : p \leq t\}$. A cut $\mathcal{C}$ is \emph{bounded} if there exists a point $t\in \mathcal{T}$ such that $\mathcal{C}\subseteq \mathcal{T}^{\leq t}$. The set of all cuts of $\mathcal{T}$ is denoted as $\mathfrak{C}(\mathcal{T})$ and the set of all bounded cuts as $\overline{\mathfrak{C}}(\mathcal{T})$.
\end{definition}
In this paper, we have taken the partially order set $\mathcal{T}$ to be Minkowski space-time, this allows us to restrict to bounded cuts. This is because every cut in Minkowski space-time is a bounded cuts: any two space-time points (even those that are unordered i.e., space-like separated) necessarily share a common causal future. Note that this is not true for a general partially ordered set $\mathcal{T}$.

\begin{definition}[Causality function \cite{Portmann2017}]
\label{definition: causality}
A function $\chi : \mathfrak{C}(\mathcal{T}) \rightarrow \mathfrak{C}(\mathcal{T})$ is a \emph{causality function} if it satisfies the following conditions:

\begin{subequations}
\begin{equation}
\label{eq: causality1}
\forall \mathcal{C},\mathcal{D}\in \mathfrak{C}(\mathcal{T}), \quad \chi(\mathcal{C}\cup\mathcal{D})=\chi(\mathcal{C})\cup\chi(\mathcal{D}),
\end{equation}

\begin{equation}
\label{eq: causality2}
\forall \mathcal{C},\mathcal{D}\in \mathfrak{C}(\mathcal{T}), \quad \mathcal{C}\subseteq\mathcal{D} \Rightarrow \chi(\mathcal{C})\subseteq\chi(\mathcal{D}),
\end{equation}

\begin{equation}
\label{eq: causality3}
\forall\mathcal{C}\in\overline{\mathfrak{C}}(\mathcal{T})\backslash \{\emptyset\}, \quad \chi(\mathcal{C})\subset \mathcal{C},
\end{equation}

\begin{equation}
\label{eq: causality4}
\forall \mathcal{C}\in \overline{\mathfrak{C}}(\mathcal{T}),\forall t\in \mathcal{C},\exists n\in \mathbb{N}, \quad t\notin \chi^n(\mathcal{C}),
\end{equation}
\end{subequations}
where $\chi^n$ denotes $n$ compositions of $\chi$ with itself, $\chi^n = \chi \circ \cdot\cdot\cdot \circ \chi$.
\end{definition}
Conditions~\ref{eq: causality1} and \ref{eq: causality2} follow from the considerations that: If the output on $\mathcal{C}$ and $\mathcal{D}$ can be computed from $\chi(\mathcal{C})$ and $\chi(\mathcal{D})$ respectively, the output on $\mathcal{C} \cup \mathcal{D}$ can be computed from $\chi(\mathcal{C}) \cup \chi(\mathcal{D})$, if $\chi(\mathcal{C})$ is needed to compute the output on $\mathcal{C}$, then certainly it is needed to compute the output on $\mathcal{D} \supseteq \mathcal{C}$. Condition~\ref{eq: causality3} is essentially the causal condition that requires that outputs of a causal box can depend only on inputs produced in its causal past and Condition~\ref{eq: causality4} is to ensure that a causal box does not produce an infinite number of messages in a finite interval of time (See \cite{Portmann2017} for details).
Definition~\ref{definition: causality} is the general definition of the causality function and it simplifies for special choices of the set $\mathcal{T}$ \cite{Portmann2017}. We are now in a position to review the formal definition of a causal box.

\subsection{General definition of a causal box}
\label{sec: cbdef}
\begin{definition}[Causal box \cite{Portmann2017}]
\label{definition: CB}
A $(d_X , d_Y)$-causal box $\Phi$ is a system with input wire $X$ and output wire $Y$ of dimension $d_X$ and $d_Y$\footnote{It is enough to define a causal box as a map from one input wire to one output wire since a single wire of dimension $d$ can always be decomposed into $n$ wires of dimensions $d_1,...,d_n$ with $d=d_1+d_2+...+d_n$ using the isomorphism of Equation(~\ref{eq: isomorphism})}, defined by a set\footnote{In general, it is modelled as a set of maps and not a single map because this allows systems to be included which produce an unbounded number of messages and are thus not well-defined as a single map on the entire set $\mathcal{T}$, but only on subsets of $\mathcal{T}$ that are upper bounded by a set of unordered points. For example \cite{Portmann2017}.} of mutually consistent (Equation~(\ref{eq: mutcons})), completely positive, trace-preserving (CPTP) maps (Equation~(\ref{eq: CB}))

\begin{equation}
\label{eq: CB}
\Phi = \{ \Phi^{\mathcal{C}}: \mathfrak{T}(\mathcal{F}_X^{\chi(\mathcal{C})}) \rightarrow \mathfrak{T}(\mathcal{F}_Y^{\mathcal{C}}) \}_{\mathcal{C}\in \overline{\mathfrak{C}}(\mathcal{T})}
\end{equation}

These maps much be such that for all $\mathcal{C}, \mathcal{D} \in \overline{\mathfrak{C}}(\mathcal{T})$ with $\mathcal{C} \subseteq \mathcal{D}$,

\begin{equation}
\label{eq: mutcons}
tr_{\mathcal{D}\backslash \mathcal{C}} \circ \Phi^{\mathcal{D}} = \Phi^{\mathcal{C}} \circ tr_{\mathcal{T}\backslash \chi(\mathcal{C})},
\end{equation}

where $\mathfrak{T}(\mathcal{F})$ denotes the set of all trace class operators on the space $\mathcal{F}$ and the causality function $\chi(.)$ satisfies all the conditions of Definition~\ref{definition: causality}. $\mathcal{F}^{\mathcal{C}}$ is the subspace of $\mathcal{F}^{\mathcal{T}}$ that contains only messages in positions $t\in \mathcal{C}\subseteq{T}$ and $tr_{\mathcal{D}\backslash\mathcal{C}}$ traces out the messages occurring at positions in $\mathcal{D} \backslash \mathcal{C}$.
\end{definition}
\par
Equation~(\ref{eq: mutcons}) can be seen as the combination of the two requirements $\Phi^{\mathcal{C}}=tr_{\mathcal{D}\backslash \mathcal{C}} \circ \Phi^{\mathcal{D}}$ and $\Phi^{\mathcal{C}}=\Phi^{\mathcal{C}} \circ tr_{\mathcal{T}\backslash \chi(\mathcal{C})}$. The first one embodies the mutual consistency requirement while the second, that of causality.
\begin{remark}
Note that Definition~\ref{definition: CB} only considers trace-preserving causal boxes. The definition can be easily generalised to non-trace preserving causal boxes or \emph{sub-normalised causal boxes} to account for post-selection. This is done in \cite{Portmann2017} by defining a suitable projector on the space of \emph{normalised causal boxes}. 
\par
Further, just like CPTP maps on quantum states, causal boxes also admit Choi-Jamio\l{}kowski and Stinespring representations, and in addition, they also admit sequence representations that describe their sequential action over subsequent, disjoint sets of $\mathcal{T}$. We refer the reader to the original paper \cite{Portmann2017} for details regarding these as they are not of particular relevance to the results of this paper.
\end{remark}
\par

\subsection{Composition of causal boxes}
\label{sec: cbcomp}
Having defined causal boxes, we are in a position to see how they can be composed. Due to Isomorphism~\ref{eq: isomorphism}, an input/output wire to a causal box of dimension $d$ can be split into sub-wires of dimensions $d_1, d_2,...,d_n$ such that $d_1+d_2+...+d_n=d$ and similarly, wires can also be combined to form a wire with dimensions equal to the sum of the dimensions of the individual wires. Taking $Ports(\Phi)$ to represent a particular partition of the input and output wires of a causal box $\Phi$ into sub-wires, arbitrary composition of causal boxes can be achieved by combining the following two steps:
\begin{enumerate}
	\item \textbf{Parallel composition: } Two causal boxes $\Phi$ and $\Psi$ can be composed in parallel to obtain a new causal box $\Gamma=\Phi\parallel\Psi$ whose input and output ports are given by the union of the input and output ports of $\Phi$ and $\Psi$ respectively.
	\item \textbf{Loops: } Selected output ports of the causal box $\Gamma$ can be connected with input ports of the same dimension to form loops.
\end{enumerate}
A classical example of composition through loops can be found in Figure~\ref{fig: cbloop}. The formal definitions of parallel composition and loop composition of causal boxes, which generalise thisintuition to the quantum case can be found in the original paper, \cite{Portmann2017} where it is also shown that causal boxes are closed under these arbitrary composition operations.

\begin{figure}[t]
\centering
	\begin{tikzpicture}[line width=0.20mm, scale=0.8, transform shape]
    		\draw[fill=white] (0,0) rectangle node{\large{$P_{CD|AB}$}}(2,2); \draw [arrows={-latex'}] (-2,1.5)--node[midway,above]{$A$}(0,1.5); \draw [arrows={-latex'}] (2,1.5)--node[midway,above]{$D$}(4,1.5); 
		\draw [arrows={-latex'}] (-2,0.5)--node[midway,above]{$B$}(0,0.5); \draw [arrows={-latex'}] (2,0.5)--node[midway,above]{$C$}(4,0.5); 
		\node[align=center] at (5.5,1) {\Large{$\mathbf{\xrightarrow{C\hookrightarrow B}}$}};
		\draw[decoration={markings, mark=at position 0.75 with {\arrowreversed{latex'}}}, postaction={decorate}](10,0) ellipse (2 and 1);
    		\draw[fill=white] (9,0) rectangle node{\large{$P_{CD|AB}$}}(11,2);\draw [arrows={-latex'}] (7,1.5)--node[midway,above]{$A$}(9,1.5); \draw [arrows={-latex'}] (11,1.5)--node[midway,above]{$D$}(13,1.5); 
		\node[align=center] at (8,0.6) {$B$}; \node[align=center] at (12,0.6) {$C$};
	\end{tikzpicture}
	\caption{\textbf{Classical example for loop composition:} A system with classical inputs $A$, $B$ and classical outputs $C$, $D$ can be described by the probability distribution $P_{CD|AB}$. The new system obtained by adding a loop from the output $C$ to input $B$ is then described by the distribution $Q_{D|A}=\sum\limits_c P_{CD|AB}(cd|ac)$ and is a valid probability distribution as long as the system obeys causality \cite{Portmann2017}.}
	\label{fig: cbloop}
\end{figure}
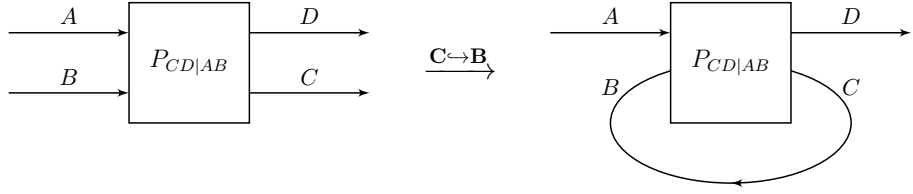

\subsection{The channel with delay as a causal box}
\label{ssec: appendix.CD}

The channel with delay was defined in Section~\ref{sssec: CD}. In this section, we show how to model the channel with delay using the causal box formalism, i.e, by defining it in terms of a set of mutually consistent maps $\{\Phi^{\mathcal{C}}\}$. Recall that a channel with delay is defined by the tuple of resources $\mathcal{CD}:=\{CD,CD_A,CD_B\}$, each of the resources $CD, CD_A$ and $CD_B$ can be equivalently described by the causal boxes $\Phi_{CD}, \Phi_{CD_A}$ and $\Phi_{CD_B}$. In the following, we consider the channel with delay resource characterised by the 4 space-time points $A\prec A'\prec B'\prec B$.

\begin{definition}[Causal box $\Phi_{CD}$ description of the channel with delay resource $CD$]
\label{definition: CDCB}
$\forall$ bounded cuts $\mathcal{C} \ni B \subseteq \mathcal{T}$ in Minkowski space $\mathcal{T}$, the causal box $\Phi_{CD} = \{ \Phi^{\mathcal{C}}_{CD}: \mathfrak{T}(\mathcal{F}_X^{\chi(\mathcal{C})}) \rightarrow \mathfrak{T}(\mathcal{F}_Y^{\mathcal{C}}) \}_{\mathcal{C}\in \overline{\mathfrak{C}}(\mathcal{T})}$ is defined by the maps $\Phi^{\mathcal{C}}_{CD}:=\mathcal{I}_{A\rightarrow B}\circ tr_{\chi(\mathcal{C})\backslash A}$, with $\mathcal{I}_{A\rightarrow B}=\mathcal{I}_{\mathcal{V}}\otimes \Big[\ket{B}\bra{A}+\ket{A}\bra{B}\Big]_{l^2(\mathcal{T})}$. $X$ and $Y$ label the input and output wires to the causal box, $\mathcal{I}_\mathcal{V}$ denotes the identity operation on the Hilbert space $\mathcal{V}$ of the quantum message, $l^2(\mathcal{T})$ is the sequence space (with bounded 2-norm) of the space-time stamps and $\chi(\mathcal{C})$ is any causality function that satisfies the conditions of Definition~\ref{definition: causality} and the condition that $B \in \mathcal{C} \Rightarrow A \in \chi(\mathcal{C})$.
\end{definition}

Similarly, the resources $CD_A$ and $CD_B$ can be defined by replacing $A$ with $A'$ and $B$ with $B'$ in Definition~\ref{definition: CDCB} respectively. Note that for any subset $\mathcal{P} \subseteq \mathcal{T}$, $\mathcal{F}^{\mathcal{T}}\cong \mathcal{F}^{\mathcal{P}}\otimes \mathcal{F}^{\tilde{\mathcal{P}}}$, where $\tilde{\mathcal{P}}=\mathcal{T}\backslash \mathcal{P}$. Further, a natural embedding of $\mathcal{F}^{\mathcal{P}}$ in $\mathcal{F}^{\mathcal{T}}$ can be obtained \cite{Portmann2017} by appending the vacuum state\footnote{$\ket{\Omega}^{\tilde{\mathcal{P}}}$ represents the one dimensional subspace of $\mathcal{F}^{\tilde{\mathcal{P}}}$ that contains the vacuum state.} to $\mathcal{F}^{\mathcal{P}}$
$$\mathcal{F}^{\mathcal{P}}\cong \mathcal{F}^{\mathcal{T}}\otimes \ket{\Omega}^{\tilde{\mathcal{P}}}\subseteq \mathcal{F}^{\mathcal{T}}$$
This allows us to equivalently view the trace $tr_{\mathcal{D}\backslash \mathcal{C}}$ for any two cuts $\mathcal{C}\subseteq \mathcal{D}$, as the operation of tracing out all the messages in space-time locations that belong to the cut $\mathcal{D}$, but not to the cut $\mathcal{C}$ and replacing all of them by the vacuum state $\ket{\Omega}$. With this, we can see that in Definition~\ref{definition: CDCB}, $tr_{\chi(\mathcal{C})\backslash A}(\rho)$ for an arbitrary input state $\rho \in \mathfrak{T}(\mathcal{F}_X^{\chi(\mathcal{C})})$ will always result in a state of the form $\sigma\otimes \ket{A}\bra{A} \otimes \ket{\Omega}\bra{\Omega}^{\tilde{A}}$ where $\sigma \in \mathcal{F}(\mathcal{V})$, which without loss of generality, we denote by $\sigma\otimes \ket{A}\bra{A}$ where it is understood that there is ``nothing" i.e., the vacuum state $\ket{\Omega}$ at all other space-time locations $\tilde{A}\in \mathcal{T}$.
\par
It is easy to verify that $\Phi_{CD}$ is indeed a causal box i.e., that it satisfies Equation~\ref{eq: mutcons}. The left hand side of the equation gives, for an arbitrary input state $\rho \in \mathfrak{T}(\mathcal{F}_X^{\chi(\mathcal{C})})$ and any cut $\mathcal{D} \ni B \supseteq \mathcal{C}$ 

\begin{equation}
\label{eq: CDCB1}
\begin{split}
    \Phi^{\mathcal{C}}_{CD}(\rho)=tr_{\mathcal{D}\backslash \mathcal{C}}\circ \Phi^{\mathcal{D}}_{CD}(\rho)=tr_{\mathcal{D}\backslash \mathcal{C}}\circ \mathcal{I}_{A\rightarrow B}\circ tr_{\chi(\mathcal{D})\backslash A} (\rho) \\
    = \begin{cases}
    tr_{\mathcal{D}\backslash \mathcal{C}}(\sigma \otimes \ket{B}\bra{B}), & A \in \chi(\mathcal{D})\\
    \ket{\Omega}\bra{\Omega}^{\mathcal{T}}, & \text{otherwise}
    \end{cases}\\
    =\begin{cases}
   \sigma \otimes \ket{B}\bra{B}, & B \in \mathcal{C}\\
    \ket{\Omega}\bra{\Omega}^{\mathcal{T}}, & \text{otherwise}
    \end{cases}
\end{split}
\end{equation}

Similarly, the right hand side of Equation~\ref{eq: mutcons} becomes
\begin{equation}
\label{eq: CDCB2}
\begin{split}
    \Phi^{\mathcal{C}}_{CD}(\rho)=\Phi^{\mathcal{C}}_{CD}\circ tr_{\mathcal{C}\backslash \chi(\mathcal{C})}(\rho)=\mathcal{I}_{A\rightarrow B}\circ tr_{\chi(\mathcal{C})\backslash A}\circ tr_{\mathcal{C}\backslash \chi(\mathcal{C})} (\rho) \\
    = \begin{cases}
   \sigma \otimes \ket{B}\bra{B}, & A \in \chi(\mathcal{C})\\
    \ket{\Omega}\bra{\Omega}^{\mathcal{T}}, & \text{otherwise}
    \end{cases}
\end{split}
\end{equation}

Since we have\footnote{Note that the implication $B \in \mathcal{C} \Rightarrow A \in \chi(\mathcal{C})$ follows from the definition of the causality function while the implication $A \in \chi(\mathcal{C}) \Rightarrow B \in \mathcal{C}$ follows from the fact that for any $\chi(\mathcal{C}) \ni A$, the causal box $\Phi_{CD}$ when acting on an arbitrary input state $\rho$, always produces an output on a cut containing $B$ (by definition).} $B \in \mathcal{C} \Leftrightarrow A \in \chi(\mathcal{C})$ by Definition~\ref{definition: CDCB},  and Equations~\ref{eq: CDCB1} and \ref{eq: CDCB2} hold for arbitrary input state $\rho$, the expressions in Equations~\ref{eq: CDCB1} and \ref{eq: CDCB2} are equal giving $tr_{\mathcal{D}\backslash \mathcal{C}}\circ \Phi^{\mathcal{D}}_{CD}=\Phi^{\mathcal{C}}_{CD}\circ tr_{\mathcal{C}\backslash \chi(\mathcal{C})}$ as required by Definition~\ref{definition: CB} of a causal box. This shows that $\Phi_{CD}$ of Definition~\ref{definition: CDCB} (and similarly $\Phi_{CD_A}$ and $\Phi_{CD_B}$) is indeed a causal box.
\begin{remark}
Note that Definition~\ref{definition: CDCB} and the fact that $\Phi_{CD}$ is a causal box imply that $\Phi_{CD}$ cannot produce any (non-vacuum) output on cuts that do not contain the point $B$. This is due to the fact that in Minkowski space, for any cut $\mathcal{C}$ with $B \notin \mathcal{C}$, we can find a cut $\mathcal{D} \supset \mathcal{C}$ containing $B$. The mutual consistency condition (Equation~\ref{eq: mutcons}) would then demand that no non-vacuum outputs are produced in the cut $\mathcal{C}$ as the only non-vacuum output in $\mathcal{D}$ will be produced at $B \notin \mathcal{C}$. Thus it is enough to define $\Phi_{CD}$ only on cuts that include $B$ as done in Definition~\ref{definition: CDCB}.
\end{remark}

%% file: appendix_proofs.tex
\subsection{Constructing \CF} 
\label{appendix: cons}

\ClaimConstructibility*

\begin{proof}

 The protocol $\Pi_{\CD \rightarrow \CF}$ (Definition~\ref{definition: protocol}) constructs $\CF^0$ from $\CD$ iff all three conditions of \fref{fig: constructibility} are satisfied. The condition of \fref{fig: constructibilitya} trivially holds. To see that the conditions in Figures~\ref{fig: constructibilityb} and \ref{fig: constructibilityc} also hold, consider the following simulators.

    $\operatorname{\sigma_A}$ is defined as follows.
\begin{enumerate}
  \item Receive the input $a$ at the space-time location $A'$ at the outer interface. If no $a$ is received, pick one uniformly at random.
  \item On receiving input $c$ at $P_2$ at the inner interface, output $b=a\oplus c$ at the outer interface at $P$.
\end{enumerate}

For the above construction of $\sigma_A$ to work, both $P_2$ and $A'$ must lie in the causal past of $P$. Since $P$ lies in the \safezone, $A'\prec P$ necessarily holds. Since there are no constraints on the space-time points at which $CF^0_A$ can produce an output, we can always make it output at a point $P_2\prec P$. 

$\operatorname{\sigma_B}$ is defined as follows.
\begin{enumerate}
  \item Receive the input $b$ at the space-time location $P$ at the outer interface. If no $b$ is received, pick one uniformly at random.
  \item On receiving input $c$ at $P_1$ at the inner interface, output $a=b\oplus c$ at the outer interface at $B'$.
\end{enumerate}

For the above construction of $\sigma_B$ to work, both $P_1$ and $P$ must lie in the causal past of $B'$. Again, $P$ being in the \safezone ensures that $P\prec B'$ necessarily holds and $P_1\prec B'$ holds since there are no restrictions on the space-time points at which $CF_B$ can produce an output.

It is easy to see that for the above mentioned constructions of the simulators $\operatorname{\sigma_A}$ and $\operatorname{\sigma_B}$, the real and ideal systems of Figures~\ref{fig: constructibilitya}-\ref{fig: constructibilityc} are perfectly indistinguishable (for any distinguisher $\mathcal{D}$) since $a$, $b$ and $c$ always satisfy the condition that any two of them sum bit-wise to the third. Further, Alice can learn the value of both bits $a$ and $b$ before Bob does but she cannot bias the value of Bob's output, $a\oplus b$ in any way. Neither can she prompt Bob to abort the protocol after learning the value of her bit $a$, because she has to send $a$ into the channel before he receives $b$ at the point $P$ (by the non empty \safezone\ condition). Hence the protocol perfectly constructs an unbiased coin flipping resource $\mathcal{CF}^0$ from a channel with delay  $\mathcal{CD}$. 
\end{proof}

\input{Fig/fig_construction_proof}

\subsection{Impossibility of \CF}
\label{appendix: impos}

\ClaimImpossibilityCF*

\input{Fig/fig_impossibility}

		
		

\begin{proof}
For the construction to be valid, all conditions of \fref{fig: impossibility} must hold. As explained in the figure caption, the first step is to combine the three conditions and use the triangle inequality to obtain \fref{fig: impossibilityd}.

Next we will show that for any causal order of the messages $c$, $c'$, $b$ and $b'$ in \fref{fig: impossibilityd}, the best possible classical, quantum or non-signalling strategy of ${\sigma}$ leads to a distinguishing advantage of at least $\frac{1}{2}(1-p)$ between $CF_B^p{\sigma}CF_A^p$ and $CF$. 
We present here only the optimal strategy---it is a straight-forward if tedious calculation to verify that all other causal orderings and possible input-output correlations in each case do not yield a lower distinguishing advantage.

The simulator's task is to ensure to the best of its capabilities that $c_o^A$ and $c_o^B$ are equal.
The causal order of the messages that provide $\operatorname{\sigma}$ with the maximum information to achieve this task is the one depicted by the \emph{directed acyclic graph} (DAG)\footnote{DAGs are widely used in the literature to represent causal structures. For classical causal structures (as is the case here,  given that the inputs and outputs to $\sigma$ are classical bits), the nodes (circles) represent random variables and the edges (arrows) represent causal influences.} in \fref{fig: causalorder}, where $\operatorname{\sigma}$ can learn the values of $c$ and $c'$ first and accordingly correlate the values of $b$ and $b'$ which are then input to $CF_A^p$ and $CF_B^p$ respectively. In this case, the best possible strategy that the simulator could adopt would be one where it produces the input-output correlations $b=b'=c$ or $b=b'=c'$ all the time.
The probability that $c_o^A$ equals $c_o^B$ for such a strategy (say, $b=b'=c$) is:
\begin{align*}
P(c_o^A=c_o^B) &= P(c_o^A=c_o^B|c=c').P(c=c')+P(c_o^A=c_o^B|c\neq c').P(c\neq c') \\
&= \frac{1}{2} + \frac{1}{2}[P(c_o^A=c_o^B|c\neq c', \text{both biased}).P(\text{both biased}) \\
&\quad+ P(c_o^A=c_o^B|c\neq c', \text{A biased}).P(\text{A biased}) \\
&\quad+ P(c_o^A=c_o^B|c\neq c', \text{B biased}).P(\text{B biased})  \\
&\quad+ P(c_o^A=c_o^B|c\neq c', \text{none biased}).P(\text{none biased})] \\
&= \frac{1}{2} + \frac{1}{2}[1.p^2 + 0.p(1-p) + 1.p(1-p) + 0.(1-p)^2] \\
&= \frac{1}{2}(1+p)
\end{align*}

\begin{figure}[tb]
	\centering
	\begin{tikzpicture}[line width=0.20mm, scale=1.7]
		\node [draw, circle, name=c1] at (0,0) {$C$}; \node [draw, circle, name=c2] at (1.5,0) {$C'$}; 
		\node [draw, circle, name=c3] at (0,1.5) {$B$}; \node [draw, circle, name=c4] at (1.5,1.5) {$B'$};  
		\draw [thick, arrows={-stealth}] (c1) -- (c3); \draw [thick, arrows={-stealth}] (c2) -- (c4); \draw [thick, arrows={-stealth}] (c1) -- (c4); \draw [thick, arrows={-stealth}] (c2) -- (c3);   
	\end{tikzpicture}
	\caption{The causal ordering of inputs and outputs of simulator $\operatorname{\sigma}$ (see Figure \ref{fig: impossibilityd}) that provide it the maximum information about the outputs $c_o$ and $c'_o$. $C$, $C'$, $B$ and $B'$ $(\in \{0,1\})$ represent the random variables of which the corresponding lower case alphabets are specific instances of. In addition, $B$ and $B'$ may causally influence each other, but this does not offer any advantage to $\operatorname{\sigma}$ because the optimal strategy is where both $b$ and $b'$ depend on $c$ (or $c'$).}
	\label{fig: causalorder}
\end{figure}

A distinguisher connected to $CF_B^p\operatorname{\sigma}CF_A^p$ or $CF$ can access the two outputs produced at the outer interfaces of these systems. If the distinguisher guesses $CF_B^p\operatorname{\sigma}CF_A^p$ every time the two outputs differ in value and $CF_B^p\operatorname{\sigma}CF_A^p$ or $CF$ with uniform probability every time the two outputs are equal, the distinguishing advantage would be:
\begin{equation}
3 \epsilon \geq d(CF_B^p\operatorname{\sigma}CF_A^p, CF) \geq P(c_o^A\neq c_o^B) = \frac{1}{2}(1-p)
\end{equation}
This distinguishing advantage $\epsilon$ is equal to zero only when $p=1$  (totally biased coin) and thus, for a non-trivial $p$, it is not possible to make the distinguishing advantage $\epsilon$ arbitrarily small.

The distinguisher used to distinguish the left and right-hand sides in \fref{fig: impossibilityd} is quite a trivial system, that only needs one bit of memory and compare the two output bits. The existence of such a distinguisher with advantage $3\epsilon$ implies that there exists another distinguisher with advantage $\eps$ that can distinguish the left and right-hand sides from either \fref{fig: impossibilitya}, \ref{fig: impossibilityb} or \ref{fig: impossibilityc}. We now go through the steps of this argument more slowly, to determine the exact complexity (both in terms of memory and computation) of the distinguisher that we have proven to exist. To construct \fref{fig: impossibilityd} from the three first conditions in \fref{fig: impossibility}, we use the following two arguments several times.

The first is the triangle inequality, namely that
\begin{equation*}
    \left.\begin{aligned}
      R \approx_\eps S\\
      S \approx_{\eps'} T
    \end{aligned}\right\} \implies R \approx_{\eps+\eps'} T.
\end{equation*}
Note that this holds for individual distinguishers, hence the contrapositive states that if there exists a distinguisher that can distinguish $R$ from $T$ with advantage $\eps+\eps'$, then \emph{exactly the same} distinguisher can distinguish either $R$ from $S$ with advantage $\eps$ or $S$ from $T$ with advantage $\eps'$.

The second generic argument---contractivity---uses the fact that for any resources $R,S$ and any other system $\alpha$,
\begin{equation*}
    R \approx_{\eps} S \implies \alpha R \approx_\eps \alpha S.
\end{equation*}
Unlike the previous argument, this one involves a change of distinguisher, namely if for some $\mathcal{D}$, $d^{\mathcal{D}}(\alpha R, \alpha S) > \eps$, then $d^{\mathcal{D}\alpha}(R,S) > \eps$, where $\mathcal{D}\alpha$ corresponds to the composition of $\mathcal{D}$ with $\alpha$.

We now start with the existence of the trivial distinguisher $\mathcal{D}$ for \fref{fig: impossibilityd} described above, and which has $d^{\mathcal{D}}(CF_B^p\operatorname{\sigma}CF_A^p, CF) > 3\epsilon$. From the triangle inequality we know that one of the three following conditions must hold
\begin{gather}
    d^{\mathcal{D}}(\Pi_A \Pi_B, CF) > \epsilon, \label{eq:impossibility.eq1} \\
    d^{\mathcal{D}}(\Pi_A\sigma_A CF_A^p,\Pi_A \Pi_B) > \epsilon, \label{eq:impossibility.eq2} \\
    d^{\mathcal{D}}(CF_B^p \sigma_B \sigma_A CF_A^p,\Pi_A\sigma_A CF_A^p) > \epsilon. \label{eq:impossibility.eq3}
\end{gather}
If it is \eqref{eq:impossibility.eq1} that holds, we are done, since we have a trivial distinguisher that can break the condition from \fref{fig: impossibilitya}. If it is either \eqref{eq:impossibility.eq2} or \eqref{eq:impossibility.eq3}, then using the contractivity rule, we find that either $\mathcal{D}\Pi_A$ can distinguish the left and right-hand sides of \fref{fig: impossibilityb} or $\sigma_A CF_A^p\mathcal{D}$ can distinguish the left and right-hand sides of \fref{fig: impossibilityc}.

Thus, both the computational requirements and memory requirements of the distinguisher are the same as the computational and memory requirements of either $\Pi_A$ or $\sigma_A CF_A^p$. 
\end{proof}

The proof of Theorem~\ref{claim: impossibility} is completely general and applies to quantum and non-signalling protocols as well. The apparent ``classicality'' of the proof is due to the fact that all inputs and outputs are classical bits as per the definition of the resources used. However, we only talk about the input-output correlations produced by the simulator $\operatorname{\sigma}$ and not the internal machinery used to produce these correlations, which could be classical, quantum or non-signalling and the impossibility holds for all classical, quantum and non-signalling strategies that $\operatorname{\sigma}$ could adopt to produce these correlations. A particular input-output correlation could be generated through many different strategies but it turns out in this particular case that there exists a simple classical strategy that perfectly produces these correlations (look at the value of $c$ and set $b=b'=c$ all the time), which is why we use correlations produced by $\operatorname{\sigma}$ and strategy adopted by $\operatorname{\sigma}$ quite interchangeably. But one must keep in mind that this in no way restricts the simulator to classical strategies.

\subsection{Impossibility of extending delays}
\label{appendix: impos_delay}

\ClaimImpossibilityExtendingCD*

\begin{proof}
Let $\mathcal{CD}^1,\dotsc,\mathcal{CD}^n$ denote the $n$ given channels with $\mathcal{CD}^i = (CD^i, CD^i_A,CD^i_B)$ and associated locations $P_i\prec P_i'\prec Q_i'\prec Q_i$.
Our goal is to  construct a channel $\mathcal{CD}'$, characterized by points $P_I \prec P_I' \prec P_F' \prec P_F$, given those channels and additional (direct) communication taking place in a space-time region $R$. 
The conditions given in \fref{fig: cdimpos} must be satisfied such that $\epsilon$ is a small, non-negative number $\forall$ distinguishers $\mathcal D \in \mathbbm D$.
In the following we write $CD = CD^1 \| \dotsb \| CD^n$ to denote the resource consisting of the  parallel composition  of the $n$ resources $CD^i$ that are available to Alice and Bob (similarly $CD_A$ and $CD_B$ for dishonest Alice and Bob respectively).

Note that for each channel with delay, there exists a converter $\delta^i_A$ such that  $\delta^i_A CD^i_A = CD^i$: this is simply a system that takes the input $a$ from Alice at position $P_i$ and outputs it at position $P_i'$.  Let  $\delta_A = \delta^1_A \| \dotsb \| \delta^n_A$ denote the parallel composition of these converters such that $\delta_A CD_A = CD$.

\begin{figure}[tbp]
	\begin{subfigure}{1.0\textwidth}
		\centering
		\begin{tikzpicture}[line width=0.20mm, scale=0.8, transform shape]
			\draw [white] (-2,0) rectangle (17,5); \draw (-0.5,0) rectangle (0.5,7); \draw (4.5,0) rectangle (5.5,7);
			\draw (2,5.5) rectangle node{$CD^1$} (3,6.5);  
			\node[align=center] at (2.5,5) {.}; \node[align=center] at (2.5,4.5) {.}; \node[align=center] at (2.5,4) {.};
			\draw (2,2.5) rectangle node{$CD^n$} (3,3.5); 
			\draw (11,2.5) rectangle node{$CD'$} (13, 4.5); 
			\draw [arrows={-stealth}] (0.5,6)--node[midway,above] {$(a_1,P_1)$}(2,6); \draw [arrows={-stealth}] (3,6)-- node[midway,above]{$(a_1, Q_1)$}(4.5,6); 
			\draw [arrows={-stealth}] (0.5,3)--node[midway,above] {$(a_n,P_n)$}(2,3); \draw [arrows={-stealth}] (3,3)-- node[midway,above]{$(a_n, Q_n)$}(4.5,3); 
			\draw [arrows={-stealth}] (0.5,0.5)--(4.5,0.5); \draw [arrows={-stealth}] (4.5,0.83)--(0.5,0.83); \draw [arrows={-stealth}] (0.5,1.16)--node[midway,above] {$(b, R)$}(4.5,1.16);
			\draw [arrows={-stealth}] (-2,3.5)--node[midway,above] {$(a, P_I)$}(-0.5,3.5); \draw [arrows={-stealth}] (5.5,3.5)--node[midway,above] {$(a, P_F)$}(7,3.5);
			\draw [arrows={-stealth}] (9.5,3.5)--node[midway,above] {$(a, P_I)$}(11,3.5); \draw [arrows={-stealth}] (13,3.5)--node[midway,above] {$(a, P_F)$}(14.5,3.5);
			\node [align=center] at (0,7.5) {$\Pi_A$};
			\node [align=center] at (5,7.5) {$\Pi_B$};
			\node [align=center] at (8.25,3.5) {\large{$\approx_{\epsilon}$}};
		\end{tikzpicture}
		\caption{$\Pi_ACD^1\parallel...\parallel CD^n\Pi_B \approx _\epsilon CD'$}
		\label{fig: cdimposa}
	\end{subfigure}
		
	\begin{subfigure}{1.0\textwidth}
		\centering
		\begin{tikzpicture}[line width=0.20mm, scale=0.8, transform shape]
		\draw [white] (-2,0) rectangle (17,5); \draw (-0.5,0) rectangle (0.5,7); 
			\draw (2,5.5) rectangle node{$CD_B^1$} (3,6.5);  
			\node[align=center] at (2.5,5) {.}; \node[align=center] at (2.5,4.5) {.}; \node[align=center] at (2.5,4) {.};
			\draw (2,2.5) rectangle node{$CD_B^n$} (3,3.5); 
			\draw (11,2.5) rectangle node{$CD_B'$} (13, 4.5); \draw (14.5,1.5) rectangle (15.5,5.5);
			\draw [arrows={-stealth}] (0.5,6)--node[midway,above] {$(a_1,P_1)$}(2,6); \draw [arrows={-stealth}] (3,6)-- node[midway,above]{$(a_1, Q_1')$}(4.5,6); 
			\draw [arrows={-stealth}] (0.5,3)--node[midway,above]{$(a_n, P_n)$}(2,3); \draw [arrows={-stealth}] (3,3)--node[midway,above]{$(a_n, Q_n')$}(4.5,3);
			\draw [arrows={-stealth}] (0.5,0.5)--(4.5,0.5); \draw [arrows={-stealth}] (4.5,0.83)--(0.5,0.83); \draw [arrows={-stealth}] (0.5,1.16)--node[midway,above]{$(b, R)$}(4.5,1.16);
			\draw [arrows={-stealth}] (-2,3.5)--node[midway,above]{$(a, P_I)$}(-0.5,3.5);
			\draw [arrows={-stealth}] (9.5,3.5)--node[midway,above]{$(a, P_I)$}(11,3.5); \draw [arrows={-stealth}] (13,3.5)--node[midway,above]{$(a, P_F')$}(14.5,3.5);
			\draw [arrows={-stealth}] (15.5,5.2)--node[pos=0.7,above]{$(a_1, Q_1')$}(16.5,5.2); \node[align=center] at (16,4.9) {.}; \node[align=center] at (16,4.6) {.}; \node[align=center] at (16,4.3) {.};
			\draw [arrows={-stealth}] (15.5,3.7)--node[pos=0.7,above]{$(a_n, Q_n')$}(16.5,3.7); \draw [arrows={-stealth}] (15.5,2.499)--node[midway,above]{$(b, R)$}(16.5,2.499); 
			\draw [arrows={-stealth}] (16.5,2.166)--(15.5,2.166); \draw [arrows={-stealth}]  (15.5,1.833)--(16.5,1.833);
			\node [align=center] at (0,7.5) {$\Pi_A$};
			\node [align=center] at (15,6) {$\sigma_B$};
			\node [align=center] at (8.25,3.5) {\large{$\approx_{\epsilon}$}};
		\end{tikzpicture}
		\caption{$\exists \operatorname{\sigma_B}$ such that $\Pi_A CD_B^1 \parallel ...\parallel CD_B^n \approx _\epsilon CD'_B\operatorname{\sigma_B}$.}
		\label{fig: cdimposb}
	\end{subfigure}
	
	\begin{subfigure}{1.0\textwidth}
		\centering
		\begin{tikzpicture}[line width=0.20mm, scale=0.8, transform shape]
		\draw [white] (-2,0) rectangle (17,5); \draw (4.5,0) rectangle (5.5,7); \draw (2,2.5) rectangle node{$CD_A^n$} (3,3.5); 
			\draw (13.5,2.5) rectangle node{$CD_A'$} (15.5, 4.5); \draw (11,1.5) rectangle (12,5.5);
			\draw (2,5.5) rectangle node{$CD_A^1$} (3,6.5);  
			\node[align=center] at (2.5,5) {.}; \node[align=center] at (2.5,4.5) {.}; \node[align=center] at (2.5,4) {.};
			\draw [arrows={-stealth}] (0.5,6)--node[midway,above] {$(a_1,P_1')$}(2,6); \draw [arrows={-stealth}] (3,6)-- node[midway,above]{$(a_1, Q_1)$}(4.5,6); 
			\draw [arrows={-stealth}] (0.5,3)--node[midway,above]{$(a_n, P_n')$}(2,3); \draw [arrows={-stealth}] (3,3)--node[midway,above]{$(a_n, Q_n)$}(4.5,3); 
			\draw [arrows={-stealth}] (0.5,0.5)--(4.5,0.5); \draw [arrows={-stealth}] (4.5,0.83)--(0.5,0.83); \draw [arrows={-stealth}] (0.5,1.16)--node[midway,above]{$(b, R)$}(4.5,1.16);
			\draw [arrows={-stealth}] (5.5,3.5)--node[midway,above]{$(a, P_F)$}(7,3.5);
			\draw [arrows={-stealth}] (12,3.5)--node[midway,above]{$(a, P_I')$}(13.5,3.5); \draw [arrows={-stealth}] (15.5,3.5)--node[midway,above]{$(a, P_F)$}(17,3.5);
			\draw [arrows={-stealth}] (10,5.2)--node[pos=0.3,above]{$(a_1, P_1')$}(11,5.2); \node[align=center] at (10.5,4.9) {.}; \node[align=center] at (10.5,4.6) {.}; \node[align=center] at (10.5,4.3) {.};
			\draw [arrows={-stealth}] (10,3.7)--node[pos=0.3,above]{$(a_n, P_n')$}(11,3.7); \draw [arrows={-stealth}] (10,2.499)--node[midway,above]{$(b, R)$}(11,2.499);
			\draw [arrows={-stealth}] (11,2.166)--(10,2.166); \draw [arrows={-stealth}] (10,1.833)--(11,1.833);
			\node [align=center] at (5,7.5) {$\Pi_B$};
			\node [align=center] at (11.5,6) {$\sigma_A$};
			\node [align=center] at (8.25,3.5) {\large{$\approx_{\epsilon}$}};
		\end{tikzpicture}
		\caption{$\exists \operatorname{\sigma_A}$ such that $CD_A^1\parallel ... \parallel CD_A^n \Pi_B \approx _\epsilon \operatorname{\sigma_A}CD'_A$.}
		\label{fig: cdimposc}
	\end{subfigure}
	\caption{Conditions for building a channel with delay $\CD'$ out of $n$ channels with delay $\CD^1, \dots, \CD^n$.}
	\label{fig: cdimpos}
\end{figure}

 From \fref{fig: cdimposc} we have
\begin{align}
    CD_A\ \Pi_B \approx_\epsilon  \sigma_A \ CD'_A &\implies
    \Pi_A\ \delta_A \ CD_A\ \Pi_B \approx_\epsilon \Pi_A \ \delta_A \  \sigma_A CD'_A \nonumber \\
    &\iff  \Pi_A\ CD\ \Pi_B \approx_\epsilon \Pi_A \ \delta_A \  \sigma_A \ CD'_A
    \label{eq:cdimpos1}
\end{align}
If we look at the right-hand side of (\ref{eq:cdimpos1}), the joint system $\Pi_A\delta_A\sigma_A$ produces an output at position $P'_I$, but nothing after. Hence, communication that does not reach $\sigma_A$ before $P'_I$ cannot influence the output and is not relevant to the output of $\Pi_A\delta_A\sigma_A$. Let $\bot_A$ denote a converter that blocks all channels $\CD^i$ with $P'_i \nprec P'_I$ and also blocks all communication in the region $R$ at points $P \nprec P'_I$. We then have $\Pi_A\bot_A\delta_A\sigma_A = \Pi_A\delta_A\sigma_A$. Combining this with \fref{fig: cdimposc}, \eqref{eq:cdimpos1}, and \fref{fig: cdimposa}, we get
\begin{equation*}
\Pi_A \bot_A \delta_A CD_A\Pi_B \approx_\epsilon \Pi_A \bot_A \delta_A \sigma_A CD'_A = \Pi_A \delta_A \sigma_A CD'_A \approx_\epsilon \Pi_A CD \Pi_B \approx_\epsilon CD',
\end{equation*}
from which we conclude that
\begin{equation} \label{eq:cdimpos2} \Pi_A \bot_A CD\Pi_B \approx_{3\epsilon} CD'. \end{equation}

We now turn our attention to \fref{fig: cdimposb}. Similarly to the argument above, we define a converter $\delta_B$ such that $CD_B\delta_B = CD$ and a converter $\bot_B$ that blocks exactly the same channels and points as $\bot_A$, but which is plugged into Bob's interface. We then get from \fref{fig: cdimposb} that
\begin{equation} \label{eq:cdimpos3} \Pi_A CD_B \delta_B \bot_B \Pi_B \approx_\epsilon CD'_B \sigma_B \delta_B \bot_B \Pi_B. \end{equation}
If we look at the left-hand side of \eqref{eq:cdimpos3}, we see that $CD_B\delta_B\bot_B = CD\bot_B = \bot_A CD$, hence it follows from \eqref{eq:cdimpos2} and \eqref{eq:cdimpos3} that 
\begin{equation} \label{eq:cdimpos4} CD'_B \sigma_B \delta_B \bot_B \Pi_B \approx_{4\epsilon} CD '. \end{equation}

\eqref{eq:cdimpos4} can only hold with $\epsilon < 1/8$ if information flows from the left interface of $CD_B'$ to the right interface of $\Pi_B$. Communication between $CD_B'$ and $\sigma_B$ only occurs in position $P'_F$, so for the message to make its way through to $\Pi_B$, there must also be communication between $\sigma_B$ and $\Pi_B$ at some point $P \succ P'_F$. The region $R$ cannot be used for this, as $P'_I \prec P'_F$ and $\bot_B$ blocks all communication after $P'_I$. The only remaining option is for there to exist a channel $\CD^i$ with $Q'_i \succ P'_F$ and which is not blocked by $\bot_B$, i.e., $P'_i \prec P'_I$. But in this case we would have $P'_i \prec P'_I \prec P'_F \prec Q_i'$, i.e., the \safezone\ of $\CD^i$ would contain the \safezone\ of $\CD'$.

To finish the proof, we still need to analyze the complexity of the distinguisher used to distinguish the real and ideal systems. The proof assumes that the protocol is secure, and then concludes that \eqref{eq:cdimpos4} must hold, which implies that the \safezone\ of the constructed channel must be contained in the \safezone\ of one of the assumed channels. Taking the contrapositive, we assume that the constructed $\CD'$ has a larger \safezone\ than the assumed channels, which implies that there exists a distinguisher that can distinguisher the left and right-hand sides of \eqref{eq:cdimpos4}, which in turn implies that there exists a distinguisher that can distinguisher the real from ideal in one of the equations from \fref{fig: cdimpos}. We will now go through the arguments of the proof to determine the complexity of this distinguisher that we have proven to exist.

The systems on the left and right-hand sides of \eqref{eq:cdimpos4} just take a message as input and output a message of the same dimension. $CD'$ performs an identity operation on the value of the message, whereas $CD'_B \sigma_B \delta_B \bot_B \Pi_B$ must trace out the input and output some fixed state, since by assumption $\CD'$ has a larger \safezone\ than the assumed channels, so there is no communication from Alice's interface to Bob's interface. If the channel is classical, an optimal system that distinguishes a fixed (possibly probabilistic) output from the identity channel, inputs a fixed message (that has low probability of being output by the channel on the left-hand side of \eqref{eq:cdimpos4}), and checks to see if the same message is output. This has probability of success at least $1/2$, and requires no memory and one equality check. If the channel is quantum, the distinguisher may perform the same (which then involves preparing one quantum state and performing a projective measurement). Alternatively, the distinguisher may input half of an EPR pair, keep the purification, and perform the projective measurement on the joint system of the output and the purification, which has a probability of success of at least $3/4$, but now involves quantum memory of the size of the message.

There are two generic arguments used in the proof to construct the distinguisher for one of the equations in \fref{fig: cdimpos} from the distinguisher for \eqref{eq:cdimpos4}. The first is the triangle inequality, namely that
\begin{equation*}
    \left.\begin{aligned}
      R \approx_\eps S\\
      S \approx_{\eps} T
    \end{aligned}\right\} \implies R \approx_{2\eps} T.
\end{equation*}
Note that this holds for individual distinguishers, hence the contrapositive states that if there exists a distinguisher that can distinguish $R$ from $T$ with advantage $2\eps$, then \emph{exactly the same} distinguisher can distinguish either $R$ from $S$ or $S$ from $T$ with advantage $\eps$.

The second generic argument uses the fact that for any resources $R,S$ and any converter $\alpha$,
\begin{equation*}
    R \approx_{\eps} S \implies \alpha R \approx_\eps \alpha S.
\end{equation*}
Unlike the previous argument, this one involves a change of distinguisher, namely if for some $\mathcal{D}$, $d^{\mathcal{D}}(\alpha R, \alpha S) > \eps$, then $d^{\mathcal{D}\alpha}(R,S) > \eps$, where $\mathcal{D}\alpha$ corresponds to the composition of $\mathcal{D}$ with $\alpha$. This was used several times in the proof with $\alpha = \Pi_A\bot_A\delta_A$, $\alpha=\Pi_A\delta_A$, and $\alpha=\delta_B\bot_B\Pi_B$. Putting this together, we prove that there exists a distinguisher than can distinguish at least one of the pairs of systems from \fref{fig: cdimpos}, and this distinguisher has the same computational requirements as either $\Pi_A$ or $\Pi_B$ along with one extra measurement needed to distinguish the left and right-hand sides of \eqref{eq:cdimpos4} (since $\delta$ and $\bot$ and forward and trace out messages, respectively, they do not perform any computation). Furthermore, if the channels are classical, then the distinguisher has the same quantum memory requirements as either $\Pi_A$ or $\Pi_B$, since $\delta$ and $\bot$ do not require any quantum memory.
\end{proof}

%% file: Fig/fig_construction_proof.tex
\begin{figure}[tbp]
	\begin{subfigure}{1.0\textwidth}
		\centering
		\begin{tikzpicture}[line width=0.20mm, scale=0.7, transform shape]
			\draw (0,0) rectangle (1,4); \draw (5,0) rectangle (6,4); 
			\draw (2.5,2.5) rectangle node{$CD$} (3.5,3.5); \draw (12,1) rectangle node{$CF$} (14, 3);
			\draw [arrows={stealth-}] (5,3)--node[midway,above]{($a$, $B$)}(3.5,3); \draw [arrows={stealth-}] (2.5,3)--node[midway,above]{($a$, $A$)}(1,3); 
			\draw [arrows={stealth-}] (1,1)--node[midway,above]{($b$, $P$)}(5,1); 
			\draw [arrows={-stealth}] (0,2)--node[midway,above]{($a\oplus b$, $P_F^A$)}(-2,2); \draw [arrows={-stealth}] (6,2)--node[midway,above]{($a\oplus b$, $P_F^B$)}(8,2); 
			\draw [arrows={-stealth}] (12,2)--node[midway,above]{($c$, $P_F^A$)}(10.5,2); \draw [arrows={-stealth}] (14,2)--node[midway,above]{($c$, $P_F^B$)}(15.5,2); 
			\node [align=center] at (0.5,4.5) {$\Pi_A$};
			\node [align=center] at (5.5,4.5) {$\Pi_B$};
			\node [align=center] at (9.25,2) {\large{$\approx_0$}};
			\node [align=center, red] at (7,-1) {$P_F^A\succ P$, $P_F^B\succ B$};
		\end{tikzpicture}
		\caption{Honest Alice and Bob: $\Pi_ACD\Pi_B\approx_0 CF$}
		\label{fig: constructibilitya}
	\end{subfigure}

	\begin{subfigure}{1.0\textwidth}
		\centering
		\begin{tikzpicture}[line width=0.20mm, scale=0.7, transform shape]
			\draw (5,0) rectangle (6,4); \draw (2.5,2.5) rectangle node{$CD_A$} (3.5,3.5); 
			\draw (15,1) rectangle node{$CF^0_A$} (17,3); \draw (12,1) rectangle node[rectangle split,rectangle split parts=2, blue] {$b=$ \nodepart{second} $a\oplus c$}(13.5,3);
			\draw [arrows={stealth-}] (5,3)--node[midway,above]{($a$, $B$)}(3.5,3); \draw [arrows={stealth-}] (2.5,3)--node[midway,above]{($a$, $A'$)}(1,3); 
			\draw [arrows={stealth-}] (1,1)--node[midway,above]{($b$, $P$)}(5,1); 
			\draw [arrows={-stealth}] (6,2)--node[midway,above]{($a\oplus b$, $P_F^B$)}(8,2); 
			\draw [arrows={stealth-}] (12,2.7)--node[midway,above]{($a$, $A'$)}(10.5,2.7); \draw [arrows={stealth-}] (10.5,1.3)--node[midway,above]{($b$, $P$)}(12,1.3); 
			\draw [arrows={-stealth}] (15,2)--node[midway,above]{($c$, $P_2$)}(13.5,2); \draw [arrows={-stealth}] (17,2)--node[midway,above]{($c$, $P_F^B$)}(18.5,2); 
			\node [align=center] at (5.5,4.5) {$\Pi_B$};
			\node [align=center] at (12.75,3.5) {$\sigma_A$};
			\node [align=center] at (9.25,2) {\large{$\approx_0$}};
			\node [align=center, red] at (10,-1) {$P_2\prec P$, $A'\prec P$};
		\end{tikzpicture}
		\caption{Dishonest Alice: $\exists \operatorname{\sigma_A}$ such that $CD_A\Pi_B\approx_0 \operatorname{\sigma_A}CF^0_A$.}
		\label{fig: constructibilityb}
	\end{subfigure}

	\begin{subfigure}{1.0\textwidth}
		\centering
		\begin{tikzpicture}[line width=0.20mm, scale=0.7, transform shape]
			\draw (0,0) rectangle (1,4); \draw (2.5,2.5) rectangle node{$CD_B$} (3.5,3.5); 
			\draw (9,1) rectangle node{$CF^0_B$} (11, 3); \draw (12.5,1) rectangle node[rectangle split,rectangle split parts=2, blue] {$a=$ \nodepart{second} $b\oplus c$}(14,3);
			\draw [arrows={stealth-}] (5,3)--node[midway,above]{($a$, $B'$)}(3.5,3); \draw [arrows={stealth-}] (2.5,3)--node[midway,above]{($a$, $A$)}(1,3); 
			\draw [arrows={stealth-}] (1,1)--node[midway,above]{($b$, $P$)}(5,1); 
			\draw [arrows={-stealth}] (0,2)--node[midway,above]{($a\oplus b$, $P_F^A$)}(-2,2); 
			\draw [arrows={stealth-}] (15.5,2.7)--node[midway,above]{($a$, $B'$)}(14,2.7); \draw [arrows={stealth-}] (14,1.3)--node[midway,above]{($b$, $P$)}(15.5,1.3); 
			\draw [arrows={-stealth}] (9,2)--node[midway,above]{($c$, $P_F^A$)}(7.5,2); \draw [arrows={-stealth}] (11,2)--node[midway,above]{($c$, $P_1$)}(12.5,2); 
			\node [align=center] at (0.5,4.5) {$\Pi_A$};
			\node [align=center] at (13.25,3.5) {$\operatorname{\sigma_B}$};
			\node [align=center] at (6.25,2) {\large{$\approx_0$}};
			\node [align=center, red] at (7,-1) {$P_1\prec B'$, $P\prec B'$};
		\end{tikzpicture}
		\caption{Dishonest Bob: $\exists \operatorname{\sigma_B}$ such that $\Pi_ACD_B\approx_0 CF^0_B\operatorname{\sigma_B}$.}
		\label{fig: constructibilityc}
	\end{subfigure}
	
	\caption{Conditions for constructibility of a fair and unbiased coin flip $\mathcal{CF}^0$ from a channel with delay $\mathcal{CD}$. Since the coin flip has $p=0$ the biasing bit that a dishonest party may input has no effect, so we do not draw it.}
	\label{fig: constructibility}
\end{figure}

%% file: Fig/fig_impossibility.tex
\begin{figure}[tbp]
	\begin{subfigure}{1.0\textwidth}
		\centering
		\begin{tikzpicture}[line width=0.20mm, scale=0.8, transform shape]
			\draw [white] (-2.1,0) rectangle (16,4.5); \draw (0,0) rectangle (1,3); \draw (3,0) rectangle (4,3); \draw (10,0.75) rectangle node{$CF$} (11.5,2.25);
			\draw [arrows={-stealth}] (1,0.75)--(3,0.75); \draw [arrows={-stealth}] (3,1.5)--(1,1.5); \draw [arrows={-stealth}] (1,2.25)--(3,2.25);
			\draw [arrows={-stealth}] (0,1.5)--node [midway,above] {\small{($d$, $P^A$)}}(-1.5,1.5); \draw [arrows={-stealth}] (4,1.5)--node[midway,above]{\small{($d$, $P^B$)}}(5.5,1.5);
			\draw [arrows={-stealth}] (10,1.5)--node[midway,above]{\small{($c$, $P^A$)}}(8.5,1.5); \draw [arrows={-stealth}] (11.5,1.5)--node[midway,above]{\small{($c$, $P^B$)}}(13,1.5);
			\node [align=center] at (0.5,3.5) {$\Pi_A$};
			\node [align=center] at (3.5,3.5) {$\Pi_B$};
			\node [align=center] at (7,1.5) {\large{$\approx_{\epsilon}$}};
		\end{tikzpicture}
		\caption{$\Pi_A\Pi_B \approx _\epsilon CF$}
		\label{fig: impossibilitya}
	\end{subfigure}\vspace{-4mm}

	\begin{subfigure}{1.0\textwidth}
		\centering
		\begin{tikzpicture}[line width=0.20mm, scale=0.8, transform shape]
			\draw [white] (-2.1,0) rectangle (16,4.5); \draw (3,0) rectangle (4,3); \draw (12,0.75) rectangle node{$CF^p_A$} (13.5,2.25); \draw (9.5,0.75) rectangle (10.5,2.25);
			\draw [arrows={-stealth}] (1,0.75)--(3,0.75); \draw [arrows={-stealth}] (3,1.5)--(1,1.5); \draw [arrows={-stealth}] (1,2.25)--(3,2.25);
			\draw [arrows={-stealth}] (4,1.5)--node[midway,above]{\small{($d$, $P^B$)}}(5.5,1.5);
			\draw [arrows={-stealth}] (12,1.95)--node[midway,above]{\small{($c'$, $P_3$)}}(10.5,1.95); \draw [arrows={-stealth}] (10.5,1.05)--node[midway,above]{\small{($b'$, $P_4$)}}(12,1.05);
			\draw [arrows={-stealth}] (13.5,1.5)--node[midway,above]{\small{($c_o^B$, $P^B$)}}(15,1.5); \draw [arrows={-stealth}] (8.5,2)--(9.5,2); \draw [arrows={-stealth}] (9.5,1.5)--(8.5,1.5); \draw [arrows={-stealth}] (8.5,1)--(9.5,1);
			\node[align=center, blue] at (0.5,1.5) {$(1)$}; \node[align=center, black!30!green] at (8,1.5) {$(2)$};
			\node [align=center] at (3.5,3.5) {$\Pi_B$};
			\node [align=center] at (10,2.75) {$\operatorname{\sigma_A}$};
			\node [align=center] at (7,1.5) {\large{$\approx_{\epsilon}$}};
		\end{tikzpicture}
		\caption{Dishonest Alice: $\exists\, \sigma_A: \quad \Pi_B \approx _\epsilon \operatorname{\sigma_A}CF_A^b$}
		\label{fig: impossibilityb}
	\end{subfigure}\vspace{-4mm}

	\begin{subfigure}{1.0\textwidth}
		\centering
		\begin{tikzpicture}[line width=0.20mm, scale=0.8, transform shape]
			\draw [white] (-2.1,0) rectangle (16,4.5); \draw (0,0) rectangle (1,3); \draw (10,0.75) rectangle node{$CF^p_B$} (11.5,2.25); \draw (13,0.75) rectangle (14,2.25);
			\draw [arrows={-stealth}] (1,0.75)--(3,0.75); \draw [arrows={-stealth}] (3,1.5)--(1,1.5); \draw [arrows={-stealth}] (1,2.25)--(3,2.25);
			\draw [arrows={-stealth}] (0,1.5)--node [midway,above] {\small{($d$, $P^A$)}}(-1.5,1.5); 
			\draw [arrows={-stealth}] (11.5,1.95)--node[midway,above]{\small{($c$, $P_1$)}}(13,1.95); \draw [arrows={-stealth}] (13,1.05)--node[midway,above]{\small{($b$, $P_2$)}}(11.5,1.05);
			\draw [arrows={-stealth}] (10,1.5)--node[midway,above]{\small{($c_o^A$, $P^A$)}}(8.5,1.5); \draw [arrows={-stealth}] (14,2)--(15,2); \draw [arrows={-stealth}] (15,1.5)--(14,1.5); \draw [arrows={-stealth}] (14,1)--(15,1);
			\node[align=center, blue] at (3.5,1.5) {$(1)$}; \node[align=center, black!30!green] at (15.5,1.5) {$(2)$};
			\node [align=center] at (0.5,3.5) {$\Pi_A$};
			\node [align=center] at (13.5,2.75) {$\operatorname{\sigma_B}$};
			\node [align=center] at (7,1.5) {\large{$\approx_{\epsilon}$}};
		\end{tikzpicture}
		\caption{Dishonest Bob: $\exists\, \sigma_B: \quad \Pi_A \approx _\epsilon CF_B^b\operatorname{\sigma_B}$}
		\label{fig: impossibilityc}
	\end{subfigure}\vspace{-4mm}

	\begin{subfigure}{1.0\textwidth}
		\centering
		\begin{tikzpicture}[line width=0.20mm, scale=0.8, transform shape]
			\draw [white] (-2,0) rectangle (19,3.5); \draw (0,0) rectangle node{$CF^p_B$}(1.5,1.5); \draw (3,0) rectangle (4,1.5); \draw (5.5,0) rectangle node{$CF^p_A$} (7,1.5); \draw (13,0) rectangle node{$CF$} (14.5,1.5);
			\draw [arrows={-stealth}] (13,0.75)--node [midway,above] {\small{($c$, $P^A$)}}(11.5,0.75); \draw [arrows={-stealth}] (14.5,0.75)--node[midway,above]{\small{($c$, $P^B$)}}(16,0.75);
			\draw [arrows={-stealth}] (0,0.75)--node [midway,above] {\small{($c_o^A$, $P^A$)}}(-1.5,0.75); \draw [arrows={-stealth}] (7,0.75)--node[midway,above]{\small{($c_o^B$, $P^B$)}}(8.5,0.75);
			\draw [arrows={-stealth}] (1.5,1.2)--node[midway,above]{\small{($c$, $P_1$)}}(3,1.2); \draw [arrows={-stealth}] (3,0.3)--node[midway,above]{\small{($b$, $P_2$)}}(1.5,0.3);
			\draw [arrows={-stealth}] (5.5,1.2)--node[midway,above]{\small{($c'$, $P_3$)}}(4,1.2); \draw [arrows={-stealth}] (4,0.3)--node[midway,above]{\small{($b'$, $P_4$)}}(5.5,0.3);
			\node [align=center] at (3.5,2) {$\operatorname{\sigma_{AB}}$};
			\node [align=center] at (10,0.75) {\large{$\approx_{3\epsilon}$}};
			\node [red, align=center] at (8.5,-1) {$P_2\prec P^A$, $P_4\prec P^B$};
		\end{tikzpicture}
		\caption{$\exists \sigma_{AB}: \quad CF_B^b\, \operatorname{\sigma}_{AB} \, CF_A^b \approx _{3\epsilon} CF$}
		\label{fig: impossibilityd}
	\end{subfigure}

	\caption{{\bf Impossibility of coin flipping: proof sketch.} For a $p$-biased Coin Flipping to be $\epsilon$-constructible solely through the exchange of messages, conditions (a)-(c) must be satisfied. 
	The composition \textcolor{blue}{$(1)$} of the system on the l.h.s.\ of (c) ($\Pi_A$) with that on the l.h.s.\ of (b) ($\Pi_B$) yields the system on the l.h.s.\ of (a) ($\Pi_A\Pi_B$) which gives the condition (d) for the corresponding right hand sides \textcolor{black!30!green}{$(2)$} with $\operatorname{\sigma_{AB}}=\operatorname{\sigma_B \ \sigma_A}$.
	To prove impossibility, we show in Appendix~\ref{appendix: impos} that for any causal order of the messages $c$, $c'$, $b$ and $b'$, the best possible classical, quantum or non-signalling strategy of $\operatorname{\sigma}$ leads to a distinguishing advantage of at least $3\epsilon=\frac{1}{2}(1-p)$ between $CF^p_B\operatorname{\sigma}CF^p_A$ and $CF$ in (d).
	Note that if the parties had access to a shared resource $\mathcal{R}$, a condition analogous to (d) could not be obtained by composing (b)  and (c), and the same impossibility proof would no longer be applicable.  
	}
	\label{fig: impossibility}
\end{figure}

%% file: appendix_extra_results.tex
\subsection{Unfair coin flipping}
\label{appendix:unfairCF}

In Section~\ref{sssec: CF}, we defined the  $p$-biased coin flipping resource tuple $\CF^p=\{CF,CF_A^p,CF_B^p\}$. Here we define another variation, the \emph{unfair} coin flipping resource tuple $\CF^{\text{uf}}$ and prove that a $1/2$-biased coin flip resource $\CF^{1/2}$ can be constructed from it. Then, by reduction, Theorem~\ref{claim: impossibility} implies the impossibility of unfair coin flipping solely through the exchange of messages.

\begin{definition}[Unfair coin flipping, $\CF^{\text{uf}}$]

An \emph{unfair coin flip} $\CF^{\text{uf}} = (CF,CF^{\text{uf}}_A,CF^{\text{uf}}_B)$ has the same resource $CF$ as $\CF^p$, and $CF^{\text{uf}}_A$ and $CF^{\text{uf}}_B$ are given by:

\begin{description}
    \item[$CF^{\text{uf}}_B$:] Bob receives a uniformly random bit $c$ at location $P_1$. At location $P_2 \succ P_1$, he can input a bit $b \in \{\perp, \overline{\perp}\}$ that may depend on the value of $c$ received at $P_1$. 
    Alice then receives message $c_o^A$ at the location $P\succ P_2$ depending on Bob's input $b$ at $P_2$: 
    if $b=\perp$, then $c_o^A = \perp$, else  $c_o^A = c$ i.e., dishonest Bob can prompt an abort ($\perp$) on Alice's interface by setting $b=\perp$.
    \item [$CF^{\text{uf}}_A$:] analogous to $CF^p_B$, with the roles reversed.
\end{description}
\end{definition}

This is illustrated in \fref{fig:unfairCF}.

	\begin{figure}[tb]
		\centering
		\begin{tikzpicture}[line width=0.20mm, scale=0.8]
			\draw (0,0) rectangle node{$CF^{\text{uf}}_B$} (1.5,1.5); \draw [arrows={-stealth}] (0,0.75)--(-1.5,0.75); \node [align=center] at (-1.65,1.05) {\small{($c_o\in \{c,\perp\}$, $P$)}};
			\draw [arrows={-stealth}] (1.5,1.2)-- node[midway,above]{\small{($c$, $P_1$)}}(3,1.2); \draw [arrows={-stealth}] (3,0.3)--(1.5,0.3); \node [align=center] at (3.2, 0.6) {\small{($b \in \{\perp, \overline{\perp}\}$, $P_2$)}};
			\node [align=center] at (-4,0.75) {\textbf{Alice}};
			\node [align=center] at (5.5,0.75) {\textbf{Bob}};
			\node [align=center, red] at (0.75,-0.75) {$P_1 \prec P_2 \prec P$};
		\end{tikzpicture}
		\caption{\label{fig:unfairCF}An unfair coin flip resource with honest Alice and dishonest Bob.}
	\end{figure}
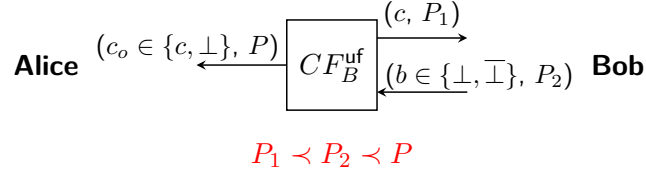

\begin{lemma}
\label{lem:unfairCF}
There exists a protocol $\Pi_{\mathcal{CF}^{\text{uf}}\rightarrow \mathcal{CF}^{1/2}}=\{\Pi_A',\Pi_B'\}$ that perfectly constructs a $1/2$-biased coin flipping resource $\mathcal{CF}^{1/2}$ from an unfair coin flipping resource $\mathcal{CF}^{\text{uf}}$.

The constructed and ideal resources are indistinguishable for any possible distinguisher (including quantum and non-signalling distinguishers). The honest protocol as well as the simulator require only elementary local operations and classical communication.
\end{lemma}

\begin{proof}

We define the honest protocol $\Pi_{\mathcal{CF}^{\text{uf}}\rightarrow \mathcal{CF}^{1/2}}=\{\Pi_A',\Pi_B'\}$ as follows:

\begin{enumerate}
	\item Receive the coin flip outcome from the corresponding interface of the unfair coin flipping resource $\mathcal{CF}^{\text{uf}}$ at the inner interface.
	\item If this outcome has a bit value (say $c$), output $c$ at the outer interface. If this outcome is an abort ($\perp$), then output $c_u=0$ or $c_u=1$ each with probability $p=1/2$ at the outer interface.
\end{enumerate}

$\Pi_{\mathcal{CF}^{\text{uf}}\rightarrow \mathcal{CF}^{1/2}}$ perfectly constructs a $1/2$-biased coin flipping resource for the following simulators (the same for $Sim_A$ and $Sim_B$):

\begin{enumerate}
	\item Receive the output bit $c'$ from the biased coin flipping resource on the inner interface and output the same bit at the outer interface.
	\item Upon receiving the additional input of $\perp$ or $\overline{\perp}$ at the outer interface, forward $b'=c'$ to the resource at the inner interface if this input is not an abort ($\overline{\perp}$) and forward $b'=\bar{c'}=c'\oplus 1$ to the resource if the input at the outer interface is an abort ($\perp$).
\end{enumerate}

One can easily verify that the real and ideal systems are identical, for convenience, we have drawn this in Figure~\ref{fig: unf-b}. 
\end{proof}

\input{Fig/fig_unfairCF}

\subsection{Abort channel}
\label{appendix:CDabort}

In Sec.~\ref{sssec: CD}, a \CD\ is defined such that once Alice inputs her message at $P$ (respectively, $P'$, if she is dishonest), Bob is guaranteed to receive it at $Q$ (or $Q'$ if he is dishonest). In this section we consider a version of a channel with delay in which Alice may additionally abort, and prevent Bob from getting her message. We call this an abort channel, and write $\CD^{\bot}$.

\begin{definition}[Abort channel, $\CD^{\bot}$]
\label{definition: CDabort}    
An abort channel $\CD^{\bot} = (CD, CD^{\bot}_A, CD_B)$ between a sender Alice and a receiver Bob is a tuple of resources characterized by five space-time locations, $P \prec  P'\prec R \prec Q' \prec Q$. $CD$ and $CD_B$ are defined identically to a standard \CD\ (Definition~\ref{definition: CD}). $CD^{\bot}_A$ is defined as follows
\begin{description}
    \item{$CD^\bot_A$:} Dishonest Alice inputs $(a, P')$. She may also input $(\bot,R)$. If she input $(\bot,R$), Bob does not receive anything. Otherwise, Bob receives $(a, Q)$.
\end{description}
\end{definition}

Nearly the same protocol as used in Theorem~\ref{claim: constructibility} can be used to construct an unfair coin flip from an abort channel.

\begin{lemma}[Construction $\CD^\bot \to \CF^{\text{uf}}$]
\label{lem:CDabort}
Given a classical abort channel $\CD^\bot$, there exists a classical protocol $\Pi_{\mathcal{CD}^\bot\rightarrow \mathcal{CF^{\text{uf}}}}=\{\Pi_A,\Pi_B\}$ that perfectly constructs an unfair coin flipping resource $\mathcal{CF}^{\text{uf}}$.

The constructed and ideal resources are indistinguishable for any possible distinguisher (including quantum and non-signalling distinguishers). The honest protocol as well as the simulator require only elementary local operations and classical communication.
\end{lemma}

\begin{proof}
The protocol is the same as the one used to construct $\CF^0$ from $\CD$, except that if Bob does not receive anything from the channel, he outputs $\bot$ instead of picking a uniform $a$ himself. The simulator $\sigma_A$ has to be changed in the same way: if it does not receive an input $(a,A')$ or if it recives $(a,A')$, but later gets an abort $\bot$ (which is now allowed by $\CD^\bot$), it notifies the resource $\CF^\bot_A$ to abort and output $\bot$ at Bob's interface. Drawing up a figure similar to \fref{fig: constructibility}, one can see that here too we have perfect security.

\end{proof}

It then follows from Theorem~\ref{claim: impossibility} that an abort channel cannot be constructed without any setup assumptions either.

\begin{corollary}[Impossibility of $\CD^\bot$]
\label{cor: noCDabort}
It is impossible to construct $\CD^\bot$, with $\epsilon<\frac{1}{12}$, between two mutually distrusting parties solely through the exchange of messages through any classical, quantum or relativistic protocol.

The distinguisher required to distinguish the real form ideal systems has the same complexity and memory requirements as the distinguisher used in Theorem~\ref{claim: impossibility} composed with the protocols used in Lemmas~\ref{lem:unfairCF} and \ref{lem:CDabort}. In particular, if these are efficient, classical and have bounded or noisy memory, then so does the distinguisher.
\end{corollary}

\begin{proof}
Lemma~\ref{lem:CDabort} constructs $\CF^{\text{uf}}$ from $\CD^\bot$, and Lemma~\ref{lem:unfairCF} constructs $\CF^{1/2}$ from $\CF^{\text{uf}}$. Thus, the impossiblity of constructing $\CF^p$ from Theorem~\ref{claim: impossibility} immediately implies the impossibility of constructing $\CD^\bot$.
\end{proof}

Finally, we can show that Theorem~\ref{claim: impos2} also holds for abort channels.

\begin{lemma}[Impossibility of extending $\CD^\bot$]
\label{lem: noExtCDabort}
Given $n$ abort channels with delay  $\mathcal{CD}_1^\bot$,...,$\mathcal{CD}_n^\bot$ between two parties, it is impossible to construct with $\epsilon \leq \frac{1}{8}$  a channel $\mathcal{CD}'^\bot$ between the two parties with a \safezone\ that is larger than the \safezone\ of all of the individual channels used.

This holds for all protocols $\Pi_A,\Pi_B$ in $\mathbbm C$, which includes inefficient and non-signalling systems. The distinguisher needed to distinguish the real from ideal system has the same complexity requirements as the protocol$\Pi_A,\Pi_B$. In particular, if it is efficient or classical, then so is the distinguisher. Furthermore, if the channels constructed and used are classical, then the distinguisher also has the same quantum memory requirements as the protocol $\Pi_A,\Pi_B$.
\end{lemma}

The proof of Lemma~\ref{lem: noExtCDabort} is identical to the proof of Theorem~\ref{claim: impos2} found in Appendix~\ref{appendix: impos_delay}, because the distinguisher used runs the honest protocol $Pi_A,Pi_B$, and $\CD$ and $\CD^\bot$ only differ on the adversarial interface (a dishonest Alice can provoke an abort). So we omit it.

%% file: Fig/fig_unfairCF.tex
\begin{figure}[tbp]
	\begin{subfigure}{1.0\textwidth}
		\centering
		\begin{tikzpicture}[line width=0.20mm, scale=0.85, transform shape]
			\draw [white] (-1,0) rectangle (18,5); \draw (0,0) rectangle (1,4); \draw (5,0) rectangle (6,4); \draw [arrows={-stealth}] (2,2)--node[midway,above]{$c$}(1,2); \draw [arrows={-stealth}] (4,2)--node[midway,above]{$c$}(5,2);
			\draw (2,1) rectangle node{$CF^{\text{uf}}$} (4,3); \draw (12,1) rectangle node{$CF^{1/2}$} (14, 3);  
			\draw [arrows={-stealth}] (0,2)--node[midway,above]{$c$}(-1,2); \draw [arrows={-stealth}] (6,2)--node[midway,above]{$c$}(7,2); 
			\draw [arrows={-stealth}] (12,2)--node[midway,above]{$c'$}(11,2); \draw [arrows={-stealth}] (14,2)--node[midway,above]{$c'$}(15,2); \draw [dashed,blue] (0,2)--(1,2); \draw [dashed, blue] (5,2)--(6,2);
			\node [align=center,blue] at (0.5,3.5) {if $\perp$,}; \draw [blue] (0.5,3) circle (0.3) node{$c_u$}; \draw [arrows={-stealth},blue] (0.5,2.7)--(0.5,2);
			\node [align=center,blue] at (5.5,3.5) {if $\perp$,}; \draw [blue] (5.5,3) circle (0.3) node{$c_u$}; \draw [arrows={-stealth},blue] (5.5,2.7)--(5.5,2);
			\node [align=center] at (0.5,4.5) {$\Pi_A$};
			\node [align=center] at (5.5,4.5) {$\Pi_B$};
			\node [align=center] at (9,2) {\large{$\approx_0$}};
		\end{tikzpicture}
		\caption{When both parties are honest, the outcomes of the unfair resource $CF^{\text{uf}}$ are never equal to $\perp$ and the protocols $\Pi_A$ and $\Pi_B$ simply forward the bit $c$ received at the inner interface to their outer interface. This is a perfect construction since the honest resources $CF^{\text{uf}}$ and $CF^{1/2}$ are the same.}
		\label{fig: unf-ba}
	\end{subfigure}
	
	\begin{subfigure}{1.0\textwidth}
		\centering
		\begin{tikzpicture}[line width=0.20mm, scale=0.85, transform shape]
			\draw [white] (-1,0) rectangle (18,5); \draw (0,0) rectangle (1,4); \draw [arrows={-stealth}] (2,2)--node[midway,above]{$\{c,\perp\}$}(1,2); 
			\draw (15,1) rectangle (16.8,3); \node [align=center] at (15.9,3.3) {$\operatorname{Sim_B}$};
			\draw [arrows={-stealth}] (4,2.5)--node[midway,above]{$c$}(5,2.5); \draw [arrows={stealth-}] (4,1.5)--(5,1.5); \node [align=center] at (4.6,1.8) {$\{\perp,\overline{\perp}\}$};
			\draw (2,1) rectangle node{$CF^{\text{uf}}$} (4,3); \draw (12,1) rectangle node{$CF^{1/2}_B$} (14, 3);  
			\draw [arrows={-stealth}] (0,2)--node[midway,above]{$\{c,c_u\}$}(-1,2);  
			\draw [arrows={-stealth}] (12,2)--node[midway,above]{$\{c',b'\}$}(11,2); \draw [arrows={-stealth}] (14,2.5)--node[midway,above]{$c'$}(15,2.5); 
			\draw [arrows={stealth-}] (14,1.5)--node[midway,above]{$b'$}(15,1.5); \draw [dashed,blue] (0,2)--(1,2); \draw [dashed,blue] (15,2.5)--(16.7,2.5); 
			\draw [arrows={-stealth}] (16.8,2.5)--node[midway,above]{$c'$}(17.8,2.5); \draw [arrows={-stealth}] (17.8,1.5)--(16.8,1.5); \node [align=center] at (17.4,1.8) {$\{\perp,\overline{\perp}\}$};
			\node [align=center,blue] at (0.5,3.5) {if $\perp$,}; \draw [blue] (0.5,3) circle (0.3) node{$c_u$}; \draw [arrows={-stealth},blue] (0.5,2.7)--(0.5,2);
			\node [align=center,blue] at (15.9, 2) {\small{$\perp \Rightarrow b'=\bar{c'}$}}; \node [align=center,blue] at (15.9, 1.5) {\small{$\overline{\perp} \Rightarrow b'=c'$}};
			\node [align=center] at (0.5,4.5) {$\Pi_A$};
			\node [align=center] at (9,2) {\large{$\approx_0$}};
		\end{tikzpicture}
		\caption{The simulator $\operatorname{Sim_B}$ for dishonest Bob simply forwards $c'$ received at its inner interface to its outer interface and sets $b'=c'$ if it receives $\overline{\perp}$ at the outer interface and $b'=\bar{c'}$ otherwise. Now the protocol $\Pi_A$ may also receive the abort input $\perp$ from the unfair CF resource, in which case it forwards the uniformly random bit $c_u$ which equals either $0$ or $1$ each with probability $1/2$ and simply forwards the input $c$ from the unfair CF resource otherwise. The construction is perfect because the probability distribution of inputs and outputs from the real system is the same as the input and output probability distribution of the ideal system and hence the two are perfectly indistinguishable. More specifically, whenever a dishonest player does not abort, the outputs at both interfaces will be equal to an independently generated, uniformly random bit (labelled as $c$ for the real system and $c'$ for the ideal system. If the dishonest player aborts, the two outputs will be equal to an independently generated, uniformly random bit ($c$ or $c'$) with a probability of $1/2$ and they will be uniformly random but completely uncorrelated ($c_u$ and $c$ for the real system and $b'$ and $c'$ for the ideal system) with a probability of $1/2$. The argument for dishonest Alice is identical.}
		\label{fig: unf-bb}
	\end{subfigure}
	\caption{Constructibility of a $1/2$-biased CF resource from an unfair CF resource. We have dropped the space-time labels corresponding to the messages to avoid unnecessary annotations, but it is easy to see that there exist space-time labels for each message involved such that the above construction is satisfied.}
	\label{fig: unf-b}
\end{figure}